%% file: Byz-Dist-Optimization.tex
\documentclass{article}

\input{headers}

\usepackage[utf8x]{inputenc}
\usepackage{pdfpages}
\usepackage{authblk}
\usepackage{nonfloat}
\usepackage{fullpage}

\date{}

\begin{document}

\title{Data Encoding for Byzantine-Resilient Distributed Optimization\thanks{This paper was presented in parts at the IEEE Allerton 2018 (as an invited talk) \cite{DataSoDi-ByzGD-Allerton18}, and ISIT 2019 \cite{DataSoDi-ByzGD-ISIT19,DataDi_ByzCD-ISIT19}.}}

\author[1]{\Large Deepesh Data}
\author[2]{\Large Linqi Song}
\author[1]{\Large Suhas Diggavi}
\affil[1]{\large University of California, Los Angeles, USA}
\affil[1] {\large \text{\{deepeshdata, suhasdiggavi\}@ucla.edu}\vspace{0.25cm}}
\affil[2]{\large City University of Hong Kong, Hong Kong}
\affil[2] {\large \text{linqi.song@cityu.edu.hk}}

\maketitle

\input{abstract}

\input{intro}

\input{problem-setting-results}

\input{related-work}

\input{solution_GD}

\input{solution_CD}

\input{extensions}

\input{experiments}

\section*{Acknowledgements}
The work of Deepesh Data and Suhas Diggavi was partially supported  by the Army Research Laboratory under Cooperative Agreement W911NF-17-2-0196, by the UC-NL grant LFR-18-548554, and by the NSF award 1740047. 
The work of Linqi Song was partially supported  by the NSF awards 1527550, 1514531, by the City University of Hong Kong grant 7200594, and by the Hong Kong RGC ECS 21212419. 
The views and conclusions contained in this document are those of the authors and should not be interpreted as representing the official policies, either expressed or implied, of the Army Research Laboratory or the U.S. Government. The U.S. Government is authorized to reproduce and distribute reprints for Government purposes notwithstanding any copyright notation here on.

\bibliographystyle{alpha}
\bibliography{reference}

\end{document}

%% file: headers.tex
\usepackage[T1]{fontenc}
\usepackage{tikz}
\usetikzlibrary{arrows}
\usetikzlibrary{calc}
\usetikzlibrary{snakes}
\usepackage{cite}
\usepackage{graphicx}
\usepackage{array}
\usepackage{mdwmath}
\usepackage{mdwtab}
\usepackage{eqparbox}
\usepackage{fixltx2e}
\usepackage{url}
\usepackage{pgf}
\usepackage{lipsum}
\usepackage{multirow}
\usepackage{subcaption}
\usepackage{amsmath,amsfonts,amssymb,amsthm}
\usepackage{xspace}
\usepackage{xcolor}

\usepackage{algorithmic}
\usepackage{algorithm}

\usepackage{bm}
\usepackage{bbm}
\usepackage{booktabs}
\usepackage{etoolbox}
\usepackage{makecell}
\usepackage[top=0.75in,bottom=0.8in,left=0.75in,right=0.75in]{geometry}
\usepackage{nicefrac}
\usepackage{textcomp}
\usepackage{stmaryrd}
\usepackage{mathrsfs}
\usepackage{paralist}
\usepackage{enumitem}
\usepackage{comment}

\usepackage[font=footnotesize,labelsep=space,labelfont=bf]{caption}

\usepackage{url}

\usepackage[pdftex,bookmarks=true,pdfstartview=FitH,colorlinks,linkcolor=blue,filecolor=blue,citecolor=blue,urlcolor=blue,pagebackref=false]{hyperref}
\makeatletter

\newcommand{\Rmnum}[1]{\expandafter\@slowromancap\romannumeral #1@}
\makeatother

\newtheorem{thm}{Theorem}
\newtheorem*{thm*}{Theorem}
\newtheorem{defn}{Definition}
\newtheorem{lem}{Lemma}
\newtheorem{claim}{Claim}
\newtheorem{remark}{Remark}

\newcommand{\namedref}[2]{\hyperref[#2]{#1~\ref*{#2}}}

\newcommand{\Algorithmref}[1]{\namedref{Algorithm}{algo:#1}}

\newcommand{\Sectionref}[1]{\namedref{Section}{sec:#1}}
\newcommand{\Subsectionref}[1]{\namedref{Section}{subsec:#1}}

\newcommand{\Theoremref}[1]{\namedref{Theorem}{thm:#1}}

\newcommand{\Lemmaref}[1]{\namedref{Lemma}{lem:#1}}
\newcommand{\Remarkref}[1]{\namedref{Remark}{remark:#1}}
\newcommand{\Claimref}[1]{\namedref{Claim}{claim:#1}}

\newcommand{\Figureref}[1]{\namedref{Figure}{fig:#1}}

\newcommand{\Footnoteref}[1]{\namedref{Footnote}{foot:#1}}
\newcommand{\Pageref}[1]{\hyperref[#1]{page~\pageref*{#1}}}

\newcommand{\eps}{\ensuremath{\epsilon}\xspace}

\newcommand{\bzero}{\ensuremath{{\bf 0}}\xspace}

\newcommand{\supp}{\ensuremath{\textsf{supp}}\xspace}
\newcommand{\bb}{\ensuremath{{\bf b}}\xspace}

\newcommand{\bc}{\ensuremath{{\bf c}}\xspace}

\newcommand{\be}{\ensuremath{{\bf e}}\xspace}
\newcommand{\bhe}{\ensuremath{\hat{\bf e}}\xspace}
\newcommand{\bte}{\ensuremath{\tilde{\bf e}}\xspace}
\newcommand{\bbf}{\ensuremath{{\bf f}}\xspace}
\newcommand{\btf}{\ensuremath{\tilde{\bf f}}\xspace}

\newcommand{\br}{\ensuremath{{\bf r}}\xspace}
\newcommand{\btr}{\ensuremath{\tilde{\bf r}}\xspace}
\newcommand{\bx}{\ensuremath{{\bf x}}\xspace}
\newcommand{\btx}{\ensuremath{\widetilde{\bf x}}\xspace}
\newcommand{\by}{\ensuremath{{\bf y}}\xspace}

\newcommand{\bz}{\ensuremath{{\bf z}}\xspace}

\newcommand{\bw}{\ensuremath{{\bf w}}\xspace}

\newcommand{\bu}{\ensuremath{{\bf u}}\xspace}
\newcommand{\bv}{\ensuremath{{\bf v}}\xspace}
\newcommand{\nr}{\ensuremath{{n_r}}\xspace}
\newcommand{\nc}{\ensuremath{{n_c}}\xspace}
\newcommand{\bA}{\ensuremath{{\bf A}}\xspace}
\newcommand{\bB}{\ensuremath{{\bf B}}\xspace}

\newcommand{\bF}{\ensuremath{{\bf F}}\xspace}

\newcommand{\bL}{\ensuremath{{\bf L}}\xspace}
\newcommand{\bR}{\ensuremath{{\bf R}}\xspace}

\newcommand{\calS}{\ensuremath{\mathcal{S}}\xspace}
\newcommand{\bS}{\ensuremath{{\bf S}}\xspace}
\newcommand{\btS}{\ensuremath{\tilde{\bf S}}\xspace}
\newcommand{\bX}{\ensuremath{{\bf X}}\xspace}
\newcommand{\btX}{\ensuremath{\widetilde{\bf X}}\xspace}

\newcommand{\B}{\ensuremath{\mathcal{B}}\xspace}
\newcommand{\C}{\ensuremath{\mathcal{C}}\xspace}
\newcommand{\I}{\ensuremath{\mathcal{I}}\xspace}

\newcommand{\U}{\ensuremath{\mathcal{U}}\xspace}

\newcommand{\N}{\ensuremath{\mathcal{N}}\xspace}
\newcommand{\setS}{\ensuremath{\mathcal{S}}\xspace}

\newcommand{\R}{\ensuremath{\mathbb{R}}\xspace}
\newcommand{\T}{\ensuremath{\mathcal{T}}\xspace}

\newcommand{\barU}{\ensuremath{\overline{\mathcal{U}}}\xspace}
\newcommand{\barfU}{\ensuremath{\overline{f(\mathcal{U})}}\xspace}

\renewcommand{\paragraph}[1]{\smallskip\noindent{\bf #1}~}

%% file: abstract.tex
\begin{abstract}
We study distributed optimization in the presence of Byzantine adversaries, where both data and computation are distributed among $m$ worker machines, $t$ of which
may be corrupt. The compromised nodes may collaboratively and arbitrarily deviate from their pre-specified programs, and a designated (master) node iteratively computes the model/parameter vector for {\em generalized linear models}.
In this work, we primarily focus on two iterative algorithms: {\em Proximal Gradient Descent} (PGD) and {\em Coordinate Descent} (CD).
Gradient descent (GD) is a special case of these algorithms. PGD is typically used in the data-parallel setting, where data is partitioned across different samples,
whereas, CD is used in the model-parallelism setting, where data is partitioned across the parameter space.

At the core of our solutions to both these algorithms is a method for Byzantine-resilient matrix-vector (MV) multiplication; and for that, we propose a method based on data encoding and error correction over real numbers to combat adversarial attacks. We can tolerate up to $t\leq \lfloor\frac{m-1}{2}\rfloor$ corrupt worker nodes, which is information-theoretically optimal. 
We give deterministic guarantees, and our method does not assume any probability distribution on the data. We develop a {\em sparse} encoding scheme which enables computationally efficient data encoding and decoding. We demonstrate a trade-off between the corruption threshold and the resource requirements (storage, computational, and communication complexity). As an example, for $t\leq\frac{m}{3}$, our scheme incurs only a {\em constant} overhead on these resources, over that required by the plain distributed PGD/CD algorithms which provide no adversarial protection. To the best of our knowledge, ours is the first paper that connects MV multiplication with CD and designs a specific encoding matrix for MV multiplication whose structure we can leverage to make CD secure against adversarial attacks.

Our encoding scheme extends {\em efficiently} to {\sf (i)} the data streaming model, in which data samples come in an online fashion and are encoded as they arrive, 
and {\sf (ii)} making {\em stochastic gradient descent} (SGD) Byzantine-resilient. In the end, we give experimental results to show the efficacy of our proposed schemes.
\end{abstract}

%% file: intro.tex
\section{Introduction}\label{sec:intro}
Map-reduce architecture \cite{DeanGhemawat} is
  implemented in many distributed learning tasks, where there is one
  designated machine (called the master) that computes the model
  iteratively, based on the inputs from the worker machines at each
  iteration, typically using descent techniques, like (proximal) gradient descent,
  coordinate descent, stochastic gradient descent, the Newton's method, etc. 
  The worker nodes perform the required computations using local data, distributed to the nodes
  \cite{zinkevich2010parallelized}. Several other architectures,
  including having no hierarchy among the nodes have been explored
  \cite{LianZhZhHsZhLi17}.
  
In several applications of distributed learning,
  including the Internet of Battlefield Things (IoBT)
  \cite{IoBTproject}, federated optimization \cite{Konecny17}, the
  recruited worker nodes might be partially trusted with their computation. Therefore, an
  important question is whether we can reliably perform distributed
  computation, taking advantage of partially trusted worker
  nodes. These Byzantine adversaries can collaborate and arbitrarily deviate from
  their pre-specified programs. The problem of distributed computation
  with Byzantine adversaries has a long history \cite{LamportShPe82},
  and there has been recent interest in applying this computational
  model to large-scale distributed learning
  \cite{BlanchardMhGuSt17,ChenWaChPa18,ChenSuXu17}.

In this paper, we study Byzantine-tolerant distributed optimization to learn a regularized {\em generalized linear model} (GLM) 
(e.g., linear/ridge regression, logistic regression, Lasso, SVM dual, constrained minimization, etc.). 
We consider two frameworks for distributed optimization: {\sf (i)} {\em data-parallelism} architecture, 
where data points are distributed across different worker nodes, and in each iteration, they all parallelly 
compute gradients on their local data and master aggregates them to update the parameter vector using gradient descent (GD) \cite{Bertsekas1989PDC,Bottou10,Dean_deeplearning12}; and
{\sf (ii)} {\em model-parallelism} architecture, where data points are partitioned across features, 
and several worker nodes work in parallel, updating different subsets of coordinates of the model/parameter vector
through {\em coordinate descent} (CD) \cite{model-parallel-11,Wright-CD-15,RichtarikTa-CD-16}.
Note that GD requires full gradients to update the parameter vector; and if full gradients are too costly to compute, we
can reduce the per-iteration cost by using CD,\footnote{Alternatively, we can also use SGD to reduce the per-iteration cost, and we give a method for making SGD Byzantine-resilient in \Subsectionref{solution_SGD}.} which also has been shown to be very effective for solving generalized linear models, and is particularly widely used for sparse logistic regression, SVM, and Lasso \cite{model-parallel-11}. 
Given its simplicity and effectiveness, CD can be chosen over GD in such applications \cite{Nesterov12}.
Computing gradients in the presence of Byzantine adversaries has been recently studied \cite{BlanchardMhGuSt17,ChenSuXu17,ChenWaChPa18,YinChRaBa18,Alistarh_Byz-SGD18,SuX_Byz19,Zeno_ByzSGD19,Bartlett-Byz_nonconvex19,NirupamVa_Byz-SGD19,Detox_ByzSGD19,RSA_Byz-SGD19,YinRobustFL19,LargrangeCoding_Yu-etal19,DataDi_Byz-SGD_Heterogeneous20,DataDi_Byz-LocalSGD_Heterogeneous20,Jaggi-Heterogeneous20},  and we discuss them in detail \Sectionref{related-work} where we also put our work in context.
However, as far as we know, making CD robust to Byzantine adversaries has not received much attention, and to the best of our knowledge, ours is the first paper that studies CD against Byzantine attacks and provides an efficient solution for that.

\subsection{Our Contributions}
We propose Byzantine-resilient
  distributed optimization algorithms both for PGD and CD based on data encoding and
  error correction (over real numbers). 
As mentioned above, there have been several papers that provide different methods for gradient computation in the presence of Byzantine adversaries, 
  however, our proposed algorithm differs from them in one or more of the following aspects: 
{\sf (i)} it does not make statistical
  assumptions on the data or Byzantine attack patterns; 
  {\sf (ii)} it can tolerate up to a constant
  fraction ($<1/2$) of the worker nodes being Byzantine, which is information-theoretically optimal; and 
  {\sf (iii)} it enables a trade-off (in terms of storage and computation/communication overhead at the master and the 
  worker nodes) with Byzantine adversary tolerance, without 
  compromising the efficiency at the master node. We give the same guarantees for CD also.

First we design a coding scheme for distributed matrix-vector (MV) multiplication, specifically, for operating in the presence of Byzantine adversaries, and use that in both our algorithms for PGD and CD to learn GLMs. Note that the connection of MV multiplication with gradient computation is straightforward and has been known for some time (see, for example, \cite{LeeLaPePaRa18,ShortDotGroverNIPS16}), however, it is not clear whether we can use MV multiplication methods for CD also.
Indeed, since each CD update has a different requirement than that of gradient computation, a general-purpose algorithm for MV multiplication may not be applicable for CD. One distinction is that in gradient computation, we only need to encode the data to compute the MV multiplication, whereas, in CD, in addition to data encoding, since workers update few coordinates of different parts of the parameter vector in parallel, we need to encode the parameter vector as well for master to be able to decode that.
In this paper, we design our encoding matrix for MV multiplication in such a way that it is sparse and has a regular structure of non-zero entries (see \eqref{eq:encoding-matrix-S_i} for the encoding matrix for any worker), which makes it applicable for CD too. This leads to efficient solutions for both PGD and CD, which are our main focus in this paper.

Inspired from the real-error correction (or sparse reconstruction) problem \cite{CandesTao05}, we develop efficient encoding/decoding procedures for MV multiplication, where we encode the data matrix and distribute it to the $m$ worker nodes, and to recover the MV product at the master, we reduce the decoding problem to the sparse reconstruction or real-error correction problem \cite{CandesTao05}.
Note that in PGD, we only need to encode the data, whereas, in CD, we also need to encode the parameter vector, and our coding scheme should facilitate the requirement that the update on a small {\em fraction} of the encoded parameter vector should affect only a small {\em fraction} of the original parameter vector. This is a non-trivial requirement, and our coding scheme for MV multiplication is designed in such a way that it supports this requirement in an efficient manner; see \Subsectionref{setting-cd} for a description on plain distributed CD, \Subsectionref{approach_CD} for our approach to making CD robust to Byzantine attacks, and \Sectionref{solution_CD} for a complete solution for Byzantine-resilient CD.
In the context of PGD/CD, for decoding, the master node processes the inputs from the worker nodes, either to compute the true gradient in the case of PGD or to facilitate the computation at the worker nodes in the case of CD. 
We take a two-round approach in each iteration of both these algorithms. 
Our main results are summarized in \Theoremref{main-result} (on \Pageref{thm:main-result}) for PGD and \Theoremref{main-result_CD} (on \Pageref{thm:main-result_CD}) for CD, and
demonstrate a trade-off between the Byzantine resilience (in terms of
the number of adversarial nodes) and the resource requirement (storage,
computational, and communication complexity). 
As an example, for $t\leq\frac{m}{3}$, our scheme incurs only a {\em constant} overhead on these resources, over that required by the plain distributed PGD and CD algorithms which provide no adversarial protection.
Our coding schemes can handle both Byzantine attacks and missing updates (\emph{e.g.,} caused by delay or asynchrony of worker nodes).
Our encoding process is also efficient. Though data encoding is a one-time
process, it has to be efficient to harness the advantage of
distributed computation. We design a sparse encoding process, based on
real-error correction, which enables efficient
encoding, and the worker nodes encode data using the
sparse structure. This allows encoding with storage redundancy\footnote{Storage redundancy is defined as the ratio of the size of the encoded matrix and the size of the raw data matrix.} of
$\frac{2m}{m-2t}$ (which is a constant, \emph{even} if $t$ is a constant ($<\frac{1}{2}$)
fraction of $m$), and a one-time total computation cost for encoding is
$O((1+2t)nd)$. Note that the time for data encoding is a factor of $(1+2t)$ (where $t$ is the corruption threshold) more than the time required for plain data distribution which is $O(nd)$, the size of the data matrix.

We extend our encoding scheme in a couple of important ways: 
first, to make the stochastic gradient descent (SGD) algorithm Byzantine-resilient without compromising much on the resource requirements; 
and second, to handle streaming data efficiently, where data points arrives one by one (and we encode them as they arrive), 
rather than being available at the beginning of the computation; we also give few more applications of our method.
For the streaming model, more specifically, our encoding requires the same amount of time, irrespective of whether we encode all the data at once, 
or we get data points one by one (or in batches) and we encode them as they arrive. 
This setting encompasses a more realistic scenario, in which we design our coding scheme with the initial set of data points and distribute the encoded data among the workers. Later on, when we get some more samples, we can easily incorporate them into our existing encoded setup.
See \Sectionref{extensions} for details on these extensions.

\subsection{Paper Organization}
We present our problem formulation, description of the plain distributed PGD and CD algorithms, 
and the high-level ideas 
of our Byzantine-resilient algorithms for both PGD and CD along-with our main results in
\Sectionref{prob-and-result}. 
We give detailed related work in \Sectionref{related-work}.  
We present our full coding schemes for MV multiplication and also for gradient computation for PGD along-with a complete analysis of their resource requirements in  \Sectionref{matrix-vector-mult}. 
In \Sectionref{solution_CD}, we provide a complete solution to CD.
In \Sectionref{extensions}, we show how our method can be
extended to SGD and to the data streaming model.
We also discuss applicability of our method to a few more important applications in that section.
In \Sectionref{experiments}, we show numerical results of our
method: we show the efficiency of our method for both gradient descent (GD) and coordinate descent (CD) by running them to solve linear regression on two
datasets (moderate and large) and plotting the running time with
varying number of corrupt worker nodes (up to <1/2 fraction).

\subsection{Notation}
We denote vectors by bold small letters (e.g., $\bx,\by,\bz$, etc.) and matrices by bold capital letters 
(e.g., $\bA,\bF,\bS,\bX$, etc.).
We denote the amount of storage required by a matrix $\bX$ by $|\bX|$. 
For any positive integer $n\in\mathbb{N}$, we denote the set $\{1,2,\hdots,n\}$ by $[n]$. 
For $n_1,n_2\in\mathbb{N}$, where $n_1\leq n_2$, we write $[n_1:n_2]$ to denote the set $\{n_1,n_1+1,\hdots,n_2\}$.
For any vector $\bu\in\R^n$ and any set $\setS\subset[n]$, we write $\bu_{\setS}$ to denote the $|\setS|$-length vector, 
which is the restriction of $\bu$ to the coordinates in the set $\setS$.
The support of a vector $\bu\in\R^n$ is defined by $\supp(\bu):=\{i\in[n]: u_i\neq0\}$.
We say that a vector $\bu\in\R^n$ is $t$-sparse if $|\supp(\bu)|\leq t$.
While stating our results, we assume that performing the basic arithmetic operations (addition, subtraction, multiplication, and division) on real numbers takes unit time.

%% file: problem-setting-results.tex
\section{Problem Setting and Our Results}\label{sec:prob-and-result}
Given a dataset consisting of $n$ labelled data points $(\bx_i,y_i)\in\R^d\times\R$, $i\in[n]$, we want to learn a model/parameter vector $\bw\in\R^d$,
which is a minimizer of the following {\em empirical risk minimization} problem:
\begin{align}\label{eq:glm}
\min_{\bw\in\R^d}\left(\bigg(\frac{1}{n}\sum_{i=1}^nf_i(\bw)\bigg)+h(\bw)\right),
\end{align}
where $f_i(\bw)$, $i=1,2,\hdots,n$, denotes the risk associated with the $i$'th data point with respect to $\bw$ and $h(\bw)$ denotes a regularizer.
We call $f(\bw):=\frac{1}{n}\sum_{i=1}^nf_i(\bw)$ the average empirical risk associated with the $n$ data points with respect to $\bw$.
Our main focus in this paper is on {\em generalized linear models} (GLM), where $f_i(\bw)=\ell(\langle\bx_i,\bw\rangle;y_i)$ for some differentiable loss function $\ell$.
Here, each $f_i:\R^d\to\R$ is differentiable, $h:\R^d\to\R$ is convex but not necessarily differentiable, and $\langle\bx_i,\bw\rangle$ is the dot product of $\bx_i$ and $\bw$. We do not necessarily need each $f_i$ to be convex, but we require $f(\bw)$ to be a convex function. Note that $f(\bw)+h(\bw)$ is a convex function. 
In the following we study different algorithms for solving \eqref{eq:glm} to learn a GLM.

\subsection{Proximal Gradient Descent}\label{subsec:setting-pgd}
We can solve \eqref{eq:glm} using {\em Proximal Gradient Descent} (PGD). This is an iterative algorithm, in which we choose an arbitrary/random initial $\bw_0\in\R^d$, and then update the parameter vector according to the following update rule:
\begin{align}\label{eq:pgd-rule} 
\bw_{t+1} = \textsf{prox}_{h,\alpha_t}(\bw_t - \alpha_t\nabla f(\bw_t)), \quad t=1,2,3,\hdots
\end{align}
where $\alpha_t$ is the step size or the learning rate at the $t$'th iteration, determining the
convergence behaviour. There are standard choices for it; see, for example, \cite[Chapter 9]{BookBoydVa04}.
For any $h$ and $\alpha$, the proximal operator $\textsf{prox}_{h,\alpha}:\R^d\to\R$ is defined as
\begin{align}\label{eq:prox-op}
\textsf{prox}_{h,\alpha}(\bw) = \arg \min_{\bz\in\R^d} \frac{1}{2\alpha}\|\bz-\bw\|_2^2 + h(\bz).
\end{align}
Observe that if $h=0$, then $\textsf{prox}_{h,\alpha}(\bw)=\bw$ for every $\bw\in\R^d$, and PGD reduces to the classical gradient descent (GD).
This encompasses several important optimization problems related to learning, for which $prox$ operator has a closed form expression; 
some of these problems are given below.
\begin{itemize}[leftmargin = *]
\item {\bf Lasso.} Here $f_i(\bw)=\frac{1}{2}(\langle\bx_i,\bw\rangle-y_i)^2$ and $h(\bw)=\lambda\|\bw\|_1$.
It turns out that $\textsf{prox}_{h,\alpha}(\bz)$ for Lasso is equal to the soft-thresholding operator $S_{\lambda\alpha}(\bz)$ \cite{proximalGD-lecture-notes}, which, for $j\in[d]$, is defined as
\[
(S_{\lambda\alpha}(\bz))_j = 
\begin{cases}
z_j+\lambda\alpha & \text{ if } z_j<-\lambda\alpha, \\
0 & \text{ if }-\lambda\alpha \leq z_j \leq \lambda\alpha, \\
z_j-\lambda\alpha & \text{ if }z_j>\lambda\alpha.
\end{cases}
\]
\item {\bf SVM dual.} Jaggi \cite{JaggiSvmLasso13} showed an equivalence between the dual formulation of Support Vector Machines (SVM) and Lasso. Hence, SVM dual is also a special case of \eqref{eq:glm}.
\item {\bf Constrained optimization.} We want to solve a constrained minimization problem $\min_{\bw\in\C}f(\bw)$, where $\C\subseteq\R^d$ is a closed, convex set. Define an indicator function $I_{\C}$ for $\C$ as follows: $I_{\C}(\bw):=0$, if $\bw\in\C$; and $I_{\C}(\bw):=\infty$, otherwise.
Now, observe the following equivalence
\[\min_{\bw\in\C}f(\bw)\iff \min_{\bw\in\R^d}f(\bw)+I_{\C}(\bw).\]
If we solve the RHS using PGD, then it can be easily verified that the corresponding proximal operator is equal to the projection operator onto the set $\C$ \cite{proximalGD-lecture-notes}. So, the proximal gradient update step is to compute the usual gradient and then project it back onto the set $\C$.
\item {\bf Logistic regression.}
Here $f_i$ is the logistic function, defined as
\[f_i(\bw) = -y_i\log\left(\frac{1}{1+e^{-u_i}}\right)-(1-y_i)\log\left(\frac{e^{-u_i}}{1+e^{-u_i}}\right),\]
where $u_i=\langle\bx_i,\bw\rangle$, and $h=0$. As noted earlier, since $h=0$, PGD reduces to GD for logistic regression.
\item {\bf Ridge regression.} Here $f_i(\bw)=\frac{1}{2}(\langle\bx_i,\bw\rangle-y_i)^2$ and $h(\bw)=\frac{\lambda}{2}\|\bw\|_2^2$.
Since $f_i$'s and $h$ are differentiable, we can alternatively solve this simply using GD.
\end{itemize}

Let $\bX\in\R^{n\times d}$ denote the data matrix, whose $i$'th row is equal to the $i$'th data point $\bx_i$. For simplicity, assume that $m$ divides $n$, and let $\bX_i$ denote the $\frac{n}{m}\times d$ matrix, whose $j$'th row is equal to $\bx_{(i-1)\frac{n}{m}+j}$.
In a distributed setup, all the data is distributed among $m$ worker machines (worker $i$ has $\bX_i$) and master updates the parameter vector using the update rule \eqref{eq:pgd-rule}. At the $t$'th iteration, master sends $\bw_t$ to all the workers; worker $i$ computes the gradient (denoted by $\nabla_i f(\bw_t)$) on its local data and sends it to the master; master aggregates all the received $m$ local gradients to obtain the global gradient 
\begin{equation}\label{eq:GD-aggregation-rule}
\nabla f(\bw_t)=\frac{1}{m}\sum_{i=1}^m\nabla_i f(\bw_t). 
\end{equation}
Now, master updates the parameter vector according to \eqref{eq:pgd-rule} and obtains $\bw_{t+1}$. Repeat the process until convergence.

\paragraph{If full gradients are too costly to compute.}
Updating the parameter vector in each iteration of PGD according to \eqref{eq:pgd-rule} requires computing full gradients. This may be 
prohibitive in large-scale applications, where each machine in a distributed framework has a lot of data, and computing full gradients at local machines may be too expensive and becomes the bottleneck.
In such scenarios, there are two alternatives to reduce this per-iteration cost:
{\sf (i)} {\em Coordinate Descent} (CD), in which we pick a few coordinates (at random), compute the partial gradient along those, and descent along those coordinates only, and {\sf (ii)} {\em Stochastic Gradient Descent} (SGD), in which we sample a data point at random, compute the gradient on that point, and descent along that direction. These are discussed in \Subsectionref{setting-cd} and \Subsectionref{solution_SGD}, respectively.

\subsection{Coordinate Descent}\label{subsec:setting-cd}
For the clear exposition of ideas, we focus on the non-regularized empirical risk minimization from \eqref{eq:glm} (i.e., taking $h=0$) for learning a {\em generalized linear model} (GLM).
This can be generalized to objectives with (non-)differentiable regularizers \cite{model-parallel-11,Shalev-ShwartzTe11}. 
Let $\bX\in\R^{n\times d}$ denote the data matrix and $\by\in\R^n$ the corresponding label vector.
To make it distinct from the last section, we denote the objective function by $\phi$ and write it as $\phi(\bX\bw;\by)$ to emphasize that we want to learn a GLM, where the objective function depends on the data points only through their inner products with the parameter vector. Formally, we want to optimize\footnote{Here we are not optimizing the {\em average} of loss functions -- since $n$ is a fixed number, this does not affect the solution space.}
\begin{align}\label{eq:problem-cd}
\min_{\bw\in\R^d}\left(\phi(\bX\bw;\by):=\sum_{i=1}^n \ell(\langle \bx_i,\bw\rangle;y_i)\right).
\end{align}
For $\U\subseteq[d]$, we write $\nabla_{\U}\phi(\bX\bw;\by)$ to denote the gradient of $\phi(\bX\bw;\by)$ with respect to $\bw_{\U}$,
where $\bw_{\U}$ denotes the $|\U|$-length vector obtained by restricting $\bw$ to the coordinates in $\U$. 
To make the notation less cluttered, let $\phi'(\bX\bw;\by)$ denote the $n$-length vector, whose $i$'th entry is equal to $\ell'(\langle \bx_i,\bw\rangle;y_i):=\frac{\partial}{\partial u}\ell(u;y_i)|_{u=\langle \bx_i,\bw\rangle}$.
Note that $\nabla\phi(\bX\bw;\by)=\bX^T\phi'(\bX\bw;\by)$ and that $\nabla_{\U}\phi(\bX\bw;\by)=\bX_{\U}^T\phi'(\bX\bw;\by)$, where $\bX_{\U}$ denotes the $n\times|\U|$ matrix obtained by restricting the column indices of $\bX$ to the elements in $\U$.

Coordinate descent (CD) is an iterative algorithm, where, in each iteration, we choose a set of coordinates and update only those 
coordinates (while keeping the other coordinates fixed). 
In distributed CD, we take advantage of the parallel architecture to improve the running time of (centralized) CD. 
In the distributed setting, we divide the data matrix vertically into $m$ parts and store the $i$'th part at the $i$'th worker node. 
Concretely, assume, for simplicity, that $m$ divides $d$.
Let $\bX=[\bX_1\ \bX_2\ \hdots\ \bX_m]$ and $\bw=[\bw_1^T\ \bw_2^T\ \hdots\ \bw_m^T]^T$, 
where each $\bX_i$ is an $n\times\frac{d}{m}$ matrix and each $\bw_i$ is a length $\frac{d}{m}$ vector.
Each worker $i$ stores $\bX_i$ and is responsible for updating (a few coordinates of) $\bw_i$ -- hence the terminology, model-parallelism. 
We store the label vector $\by$ at the master node. 
In coordinate descent, since we update only a few coordinates in each round, there are a few options on how to update these coordinates in a distributed manner: 

\paragraph{Subset of workers:}
Master picks a subset $\calS\subset[m]$ of workers and asks them to update their $\bw_i$'s \cite{RichtarikTa-CD-16}. 
This may not be good in the adversarial setting, because if only a small subset of workers are updating their parameters, 
the adversary can corrupt those workers and disrupt the computation.

\paragraph{Subset of coordinates for all workers:}
All the worker nodes update only a subset of the coordinates of their local parameter vector $\bw_i$'s.
Master can (deterministically or randomly) pick a subset $\U$ (which may or may not be different for all workers) of $f\leq d/m$ coordinates and asks each worker to updates only those coordinates. If master picks $\U$ deterministically, it can cycle through and update all coordinates of the parameter vector in $\lceil d/mf\rceil$ iterations. 

In \Algorithmref{distr-coor-desc-algo}, we give the distributed CD algorithm with the second approach, where all worker nodes update the coordinates of their local parameter vectors for a single subset $\U$. We will adopt this approach in our method to make the distributed CD Byzantine-resilient.
Let $r=\frac{d}{m}$. For any $i\in[m]$, let $\bw_i=[w_{i1}\ w_{i2}\hdots w_{ir}]^T$ and $\bX_i=[\bX_{i1}\ \bX_{i2}\hdots\bX_{ir}]$, where $\bX_{ij}$ is the $j$'th column of $\bX_i$.
For any $i\in[m]$ and $\U\subseteq[r]$, let $\bw_{i\U}$ denote the $|\U|$-length vector that is obtained from $\bw_i$ by
restricting its entries to the coordinates in $\U$; similarly, let $\bX_{i\U}$ denote the $n\times|\U|$ matrix obtained by restricting the column indices of $\bX_i$ to the elements in $\U$.

\begin{algorithm}[t]
   \caption{Distributed Coordinate Descent}\label{algo:distr-coor-desc-algo}
\begin{algorithmic}[1]
   \STATE {\bf Initialize.} Each worker $i\in[m]$ starts with an arbitrary/random $\bw_i\in\R^r$, where $r=\frac{d}{m}$ and, for simplicity, we assume that $m$ divides $d$.
   \WHILE{(until the stopping criteria at master is not satisfied)}
  \STATE \textbf{On each worker $i\in[m]$, do in parallel:}
   \STATE Worker $i$ computes $\bX_i\bw_i$ and sends it to the master node.\footnotemark
   \STATE Worker $i$ receives $(\U\subseteq[r],\phi'(\bX\bw;\by))$ from the master node.
   \STATE Worker $i$ updates its local parameter vector as (where $\nabla_{i\U}\phi(\bX\bw;\by)=\bX_{i\U}^T\phi'(\bX\bw;\by)$)
\begin{align}\label{eq:dist-coor-dist-update}
\bw_{i\U}\leftarrow \bw_{i\U} - \alpha \nabla_{i\U}\phi(\bX\bw;\by)
\end{align}
while keeping the other coordinates of $\bw_i$ unchanged, 
and sends the updated $\bw_i$ to the master. 
 
   \STATE \textbf{At Master:}
   \STATE Master receives $\{\bX_i\bw_i\}_{i\in[m]}$ from the $m$ workers.
   \STATE Master first computes $\bX\bw=\sum_{i=1}^m\bX_i\bw_i$ and then computes $\phi'(\bX\bw;\by)$. 
   \STATE Master picks $\U\subseteq[r]$ (where $\U$ can be picked either randomly or in a round-robin fashion) and sends $(\U\subseteq[r],\phi'(\bX\bw;\by))$ to all workers.
   \ENDWHILE  
\end{algorithmic}
\end{algorithm}
\footnotetext{After the 1st iteration, worker $i$ need not multiply $\bX_i$ with $\bw_i$ to obtain $\bX_i\bw_i$ in every iteration; as only a few coordinates of $\bw_i$ are updated, it only needs to multiply those columns of $\bX_i$ that corresponds to the updated coordinates of $\bw_i$.}

In \Algorithmref{distr-coor-desc-algo}, for each worker $i$ to update $\bw_i$ according to \eqref{eq:dist-coor-dist-update}, where the partial gradient of $\phi$ with respect to $\bw_{i\U}$ is equal to $\nabla_{i\U}\phi(\bX\bw;\by)$ = $\bX_{i\U}^T\phi'(\sum_{j=1}^m\bX_j\bw_j;\by)$ and worker $i$ has only $(\bX_i,\bw_i)$, 
every other worker $j$ sends $\bX_j\bw_j$ to the master, who computes $\phi'(\sum_{j=1}^m\bX_j\bw_j;\by)$\footnote{\label{foot:storing-labels}Note that even after computing $\bX\bw$, master needs access to the labels $y_i, i=1,2,\hdots,n$ to compute $\phi'(\bX\bw;\by)$. Since $\by\in\R^n$ is just a vector, we can either store that at master, or, alternatively, we can encode $\by$ distributedly at the workers and master can recover that using the method developed in \Sectionref{matrix-vector-mult} for Byzantine-resilient distributed matrix-vector multiplication, where the matrix is an identity matrix and vector is equal to $\by$.}  and sends it back to all the workers.
Observe that, even if one worker is corrupt, it can send an adversarially chosen vector to make the computation at the master deviate arbitrarily from the desired computation, which may adversely affect the update at all the worker nodes subsequently.\footnote{Specifically, suppose the $i$'th worker is corrupt and the adversary wants master to compute $\phi'(\bX\bw+\be;\by)$ for any arbitrary vector $\be\in\R^n$ of its choice, then the $i$'th worker can send $\bX_i\bw_i+\be$ to the master.}
Similarly, corrupt workers can send adversarially chosen information to affect the stopping criterion.

\subsection{Adversary Model}\label{subsec:adversary-model}
We want to perform the distributed computation described in \Subsectionref{setting-pgd} and \Subsectionref{setting-cd} under adversarial attacks, where the corrupt nodes may provide erroneous vectors to the master node. Our adversarial model is described next.

In our adversarial model, the adversary can
  corrupt at most $t<\frac{m}{2}$ worker nodes\footnote{Our results also apply to a slightly
    \emph{different} adversarial model, where the adversary can
    adaptively \emph{choose} which of the $t$ worker nodes to attack
    at each iteration. However, in this model, the adversary cannot
    modify the local stored data of the attacked node, as otherwise,
    over time, it can corrupt all the data, making any defense
    impossible.}, 
    and the compromised nodes may collaborate and 
  arbitrarily deviate from their pre-specified programs. If a worker
is corrupt, then instead of sending the true vector, it may send an
arbitrary vector to disrupt the computation. We refer to the corrupt
nodes as erroneous or under the Byzantine attack.
We can also handle asynchronous
updates, by dropping the straggling nodes beyond a specified delay,
and still compute the correct gradient due to encoding. Therefore we
treat updates from these nodes as being ``erased''.  We refer to these
as erasures/stragglers.  For every worker $i$ that sends a message to
the master, we can assume, without loss of generality, that the master
receives $\bu_i+\be_i$, where $\bu_i$ is the true vector and $\be_i$ is the error vector, where $\be_i=\bzero$ if the $i$'th node is
honest, otherwise can be arbitrary.
We assume that at most $t$ nodes can be adversarially corrupt and at most $s$ nodes can be stragglers, where $s$ and $t$ are some constants less than $\frac{1}{2}$ that we will decide later.
Note that the master node does not know which $t$ worker nodes are corrupted (which makes this problem non-trivial to solve), 
but knows $t$. 
We propose a method that mitigates the effects of both of these anomalies.
\begin{remark}
  A well-studied problem is that of asynchronous distributed
  optimization, where the workers can have different delays in updates
  \cite{DeanTailAtScale}. One mechanism to deal with this is to wait
  for a subset of responses, before proceeding to the next iteration,
  treating the others as missing (or erasures) \cite{KarakusSuDiYi17}.
  Byzantine attacks are quite distinct from such erasures, as the
  adversary can report wrong local gradients, requiring the master
  node to create mechanisms to overcome such attacks. If the master
  node simply aggregates the collected updates as in
  \eqref{eq:GD-aggregation-rule}, the computed gradient could be
  arbitrarily far away from the true one, even with a single adversary
  \cite{MhamdiGuRo18}.
\end{remark}

\subsection{Our Approach to Gradient Computation}\label{subsec:approach_GD}
Recall that $f_i(\bw)=\ell(\langle \bx_i,\bw\rangle;y_i)$ for some differentiable loss function $\ell$, and the gradient of $f_i$ at $\bw$ is equal to $\nabla f_i(\bw)=(\bx_i)^T\ell'(\langle \bx_i,\bw\rangle;y_i)$, where $\ell'(\langle \bx_i,\bw\rangle;y_i):=\frac{\partial}{\partial u}\ell(u;y_i)|_{u=\langle \bx_i,\bw\rangle}$. Note that $\nabla f_i(\bw)\in\R^d$ is a column vector.
Let $f'(\bw)$ denote the $n$-length vector whose $i$'th entry is equal to $\ell'(\langle \bx_i,\bw\rangle;y_i)$. 
With this notation, since $f(\bw)=\frac{1}{n}\sum_{i=1}^nf_i(\bw)$, we have $\nabla f(\bw)=\frac{1}{n}\bX^Tf'(\bw)$.
Since $n$ is a constant, it is enough to compute $\bX^Tf'(\bw)$. So, for simplicity, in the rest of the paper we write
\begin{equation}\label{eq:gradient-2-step}
\nabla f(\bw)=\bX^Tf'(\bw), \quad \forall\bw\in\R^d.
\end{equation}

A natural approach to computing the gradient $\nabla f(\bw)$ is to compute it in two rounds: 
{\sf (i)} compute $f'(\bw)$ in the 1st round by first multiplying $\bX$ with $\bw$ and then master locally computes $f'(\bw)$ from $\bX\bw$ 
(master can do this locally, because $\bX\bw$ is an $n$-dimensional vector whose $i$'th entry is equal to $\langle \bx_i,\bw\rangle$ and $(f'(\bw))_i=\ell'(\langle \bx_i,\bw\rangle;y_i)$);\footnote{Note that even after computing $\bX\bw$, master needs access to the labels $y_i, i=1,2,\hdots,n$ to compute $f'(\bw)$. See \Footnoteref{storing-labels} for a discussion on how master can get access to the labels.}  and then 
{\sf (ii)} compute $\nabla f(\bw) = \bX^Tf'(\bw)$ in the 2nd round by multiplying $\bX^T$ with $f'(\bw)$.
So, the task of each gradient computation reduces to two matrix-vector (MV) multiplications, where the matrices are fixed and vectors may be different each time. 
To combat against the adversarial worker nodes, we do both of these MV multiplications 
using data encoding and real-error correction; see \Figureref{problem-setup} on \Pageref{fig:problem-setup} for a pictorial description of our approach.

A two-round approach for gradient computation has been proposed for straggler mitigation
in \cite{LeeLaPePaRa18}, but our method for MV multiplication differs
from that fundamentally, as we have to provide adversarial protection. 
Note that in the case of stragglers/erasures we know who the straggling nodes are, but this information is not known in the case of adversarial nodes, and master needs to decode without this information in the context of Byzantine adversaries. This is slightly different from the standard error correcting codes (over finite fields) as the matrix entries in machine learning applications are from reals. In this case, we use ideas from real-error correction (or sparse reconstruction) from the compressive sensing literature \cite{CandesTao05}, and using which we develop an efficient decoding at master, which also gives rise to our sparse encoding matrix; see \Sectionref{matrix-vector-mult} for more details. For decoding efficiently, we crucially leverage the block error pattern and design a decoding method at master, which, interestingly, requires {\em just one} application of the sparse recovery method on a vector of size $m$, the number of workers, which may be much smaller than the data dimensions $n$ and $d$, thereby making the decoding computationally efficient. Our encoding matrix (given in \eqref{eq:encoding-matrix-S_i}, designed for MV multiplication) is very sparse and has a regular pattern of non-zero entries, which also makes it applicable for making coordinate-descent (CD) Byzantine-resilient. We emphasize that a general-purpose code for MV multiplication may not be applicable for CD, as each CD iteration requires updating only a few coordinates of the parameter vector, which makes it fundamentally different (and arguably more complicated to robustify) than GD iterations; see \Subsectionref{related-work_CD} and \Sectionref{solution_CD} for more details.
Since iterative algorithms (such as GD and CD) require repeated parameter updates,
it is crucial to have a method that has low computational complexity, both at the worker nodes as well as at the master node, and our coding solutions for both GD and CD achieve that, in addition to being highly storage efficient; see \Theoremref{main-result} for GD and \Theoremref{main-result_CD} for CD.

Coming back to our two-round approach for gradient computations using MV multiplications, for the 1st round, we encode $\bX$ using a sparse encoding matrix $\bS^{(1)}=[(\bS_1^{(1)})^T,\hdots,(\bS_m^{(1)})^T]^T$ and store $\bS_i^{(1)}\bX$ at the $i$'th worker node; and
for the 2nd round, we encode $\bX^T$ using another sparse encoding matrix $\bS^{(2)}=[(\bS_1^{(2)})^T,\hdots,(\bS_m^{(2)})^T]^T$, and store $\bS_i^{(2)}\bX^T$ at the $i$'th worker node.
Now, in the 1st round of the gradient computation at $\bw$, the master node broadcasts $\bw$ and the $i$'th worker node replies with $\bS_i^{(1)}\bX\bw$ (a corrupt worker may report an arbitrary vector); upon receiving all the vectors, the master node applies error-correction procedure to recover $\bX\bw$ and then locally computes $f'(\bw)$ as described above. In the 2nd round, the master node broadcasts $f'(\bw)$ and similarly can recover $\bX^Tf'(\bw)$ (which is equal to the gradient) at the end of the 2nd round.
So, it suffices to devise a method for multiplying a vector $\bv$ to a fixed matrix $\bA$ in a distributed and adversarial setting. Since this is a linear operation, we can apply error correcting codes over real numbers to perform this task. We describe it briefly below.

\paragraph{A trivial approach.}\label{trivial-approach} 
Take a generator matrix ${\bf G}$ of any real-error correcting linear code. Encode $\bA$ as $\bA^T{\bf G}=:{\bf B}$.
Divide the columns of $\bB$ into $m$ groups as $\bB=[\bB_1\ \bB_2 \hdots \bB_m]$, where worker $i$ stores $\bB_i$.
Master broadcasts $\bv$ and each worker $i$ responds with $\bv_T\bB_i+\be_i^T$, where $\be_i=\bzero$ if the $i$'th worker is honest, otherwise can be arbitrary. Note that at most $t$ of the $\be_i$'s can be non-zero. 
Responses from the workers can be combined as $\bv^T{\bf B}+\be^T$. Since every row of ${\bf B}$ is a codeword, $\bv^T{\bf B}=\bv^T\bA^T{\bf G}$ is also a codeword. Therefore, one can take any off-the-shelf decoding algorithm for the code whose generator matrix is ${\bf G}$ and obtain $\bv^T\bA^T$.
For example, we can use the Reed-Solomon codes (over real numbers) for this purpose, which only incurs a constant storage overhead and tolerates optimal number of corruptions (up to $<\frac{1}{2}$).
Note that we need fast decoding, as it is performed in every iteration of the gradient computation by the master. 
As far as we know, any off-the-shelf decoding algorithm ``over real numbers'' requires at least a quadratic computational complexity, which leads to $\Omega(n^2 + d^2)$ decoding complexity per gradient computation, which could be impractical. 

The trivial scheme does not exploit the block error pattern which we crucially exploit in our coding scheme to give a $\sim O((n+d)m)$ time decoding per gradient computation, which could be a significant improvement over the trivial scheme, since $m$ typically is much smaller than $n$ and $d$ for large-scale problems.
In fact, our coding scheme enables a trade-off (in terms of storage and computation/communication overhead at the master and 
  the worker nodes) with Byzantine adversary tolerance, without
  compromising the efficiency at the master node.
We also want encoding to be efficient (otherwise it defeats the purpose of data encoding) and our sparse encoding matrix achieves that. 
Our main result for the Byzantine-resilient distributed gradient computation is as follows, which is proved in \Sectionref{matrix-vector-mult}:
\begin{thm}[Gradient Computation]\label{thm:main-result}
Let $\bX\in\R^{n\times d}$ denote the data matrix. Let $m$ denote the total number of worker nodes.
We can compute the gradient exactly in a distributed manner in the presence of $t$ corrupt worker nodes and $s$ stragglers, with the following guarantees, where $\epsilon>0$ is a free parameter.
\begin{itemize}
\item $(s+t)\leq\left\lfloor\frac{\epsilon}{1+\epsilon}\cdot \frac{m}{2}\right\rfloor$. 
\item Total storage requirement is roughly $2(1+\epsilon)|\bX|$. 
\item Computational complexity for each gradient computation: 
\begin{itemize}
\item at each worker node is $O((1+\epsilon)\frac{nd}{m})$.
\item at the master node is $O((1+\epsilon)(n+d)m)$.
\end{itemize}
\item Communication complexity for each gradient computation:
\begin{itemize}
\item each worker sends $\left((1+\epsilon)\frac{n+d}{m}\right)$ real numbers.
\item master broadcasts $(n+d)$ real numbers.
\end{itemize}
\item Total encoding time is $O\left(nd\left(\frac{\epsilon}{1+\epsilon}m+1\right)\right)$.
\end{itemize}
\end{thm}

\begin{remark}\label{remark:errors-erasures}
The statement of \Theoremref{main-result} allows for any $s$ and $t$ 
as long as $(s+t)\leq\left\lfloor\frac{\epsilon}{1+\epsilon}\cdot
\frac{m}{2}\right\rfloor$. As we are handling
  both erasures and errors in the same way\footnote{When there are
    \emph{only stragglers}, one can design an encoding scheme where
    \emph{both} the master and the worker nodes operate oblivious to
    encoding, while solving a slightly altered optimization problem
    \cite{KarakusSuDiYi17}, in which gradients are computed {\em approximately}, 
    leading to more efficient straggler-tolerant GD.} the corruption threshold does not have to
  handle $s$ and $t$ separately. To simplify the discussion, for the
  rest of the paper, we consider only Byzantine corruption, and denote the
  corrupted set by $\I\subset[m]$ with $|\I|\leq t$, with the
  understanding that this can also work with stragglers.

In \Theoremref{main-result}, $\eps$ is a design choice and a free parameter that can take any value in the interval $[0,m-1]$, where $\eps=0$ implies no corruption and $\eps=m-1$ implies that corruption threshold $t$ can be anything up to $\frac{m-1}{2}$. If we want to tolerate $t$ corrupt workers, then $\eps$ must satisfy $\eps\geq\frac{2t}{m-2t}$.\footnote{We could have written everything in terms of $t,m,n,d$, but we chose to introduce another variable $\eps$ which, in our opinion, clearly brings out the tradeoff between the corruption threshold and the resource requirements without cluttering the expressions.}
  
    \end{remark}

\begin{remark}[Comparison with the plain distributed PGD]\label{remark:comparison_pgd}
We compare the resource requirements of our method with the plain distributed PGD (which provides {\em no adversarial protection}), 
where all the data points are evenly distributed among the $m$ workers. In each iteration, master sends the parameter vector $\bw$ to all the workers; upon receiving $\bw$, all workers compute the gradients on their local data in $O(\frac{nd}{m})$ time (per worker) and send them to the master; master aggregates them in $O(md)$ time to obtain the global gradient and then updates the parameter vector using \eqref{eq:pgd-rule}.

In our scheme {\sf (i)} the total storage requirement is $O(1+\epsilon)$ factor more;\footnote{For example, by taking $\eps=2$, our method can tolerate $m/3$ corrupt worker nodes. So, we can tolerate linear corruption with a {\em constant} overhead in the resource requirement, compared to the plain distributed gradient computation which does not provide any adversarial protection.} (see also \Remarkref{parameters}) {\sf (ii)} the amount of computation at each worker node is $O(1+\epsilon)$ factor more; {\sf (iii)} the amount of computation at the master node is $O((1+\epsilon)(1+\frac{n}{d}))$ factor more, which is comparable in cases where $n$ is not much bigger than $d$; {\sf (iv)} master broadcasts $(1+\frac{n}{d})$ factor more data, which is comparable if $n$ is not much bigger than $d$; and {\sf (v)} each worker sends $O\left((1+\epsilon)\frac{1+\nicefrac{n}{d}}{m}\right)$ factor more data, which is 
$O(1+\epsilon)$ -- a constant factor -- as long as $n=O(dm)$.
\end{remark}

\begin{remark}\label{remark:parameters}
Let $m$ be an even number. Note that we can get the corruption threshold $t$ to be any number less than $m/2$, but at the expense of increased storage and computation. For any $\delta>0$, if we want to get $\delta$ close to m/2, i.e., $t=m/2-\delta$, then we must have $(1+\epsilon)\geq m/2\delta$. In particular, at $\epsilon=2$, we can tolerate up to $m/3$ corrupt nodes, with constant overhead in the total storage as well as on the computational complexity. 

Note that when $\delta$ is a constant, i.e., $t$ is close to $\frac{m-1}{2}$, then $\eps$ grows linearly with $m$; 
for example, if $t=\frac{m-1}{2}$, then $\eps=m-1$.
In this case, our storage redundancy factor is $O(m)$. In contrast, the trivial scheme (see ``trivial approach'' on \Pageref{trivial-approach}) 
does better in this regime and has only a constant storage overhead, but at the expense of an increased decoding complexity at the master, 
which is at least quadratic in the problem dimensions $d$ and $n$, whereas, our decoding complexity at the master always scales linearly with $d$ and $n$. If we always want a constant storage redundancy for all values of the corruption threshold $t$, we can use our coding scheme if $t\leq c\cdot\frac{m-1}{2}$, where $c<1$ is a constant, and use the trivial scheme if $t$ is close to $\frac{m-1}{2}$. 

Our encoding is also efficient and requires $O\left(nd\left(\frac{\epsilon}{1+\epsilon}m+1\right)\right)$ time.
Note that $O(nd)$ is equal to the time required for distributing the data matrix $\bX$ among $m$ workers (for running the distributed gradient descent algorithms {\em without} the adversary); and the encoding time in our scheme (which results in an encoded matrix that provides Byzantine-resiliency) is a factor of $(2t+1)$ more.
\end{remark}

\begin{remark}\label{remark:optimal-threshold}
Our scheme is not only efficient (both in terms of computational complexity and storage requirement),
but it can also tolerate up to $\lfloor\frac{m-1}{2}\rfloor$ corrupt worker nodes (by taking $\epsilon=m-1$ in \Theoremref{main-result}).
It is not hard to prove that this bound is information-theoretically optimal, i.e., no algorithm can tolerate $\lceil\frac{m}{2}\rceil$ corrupt worker nodes,
and at the same time correctly computes the gradient.
\end{remark}

\subsection{Our Approach to Coordinate Descent}\label{subsec:approach_CD}
We use data encoding and add redundancy to enlarge the parameter space.
Specifically, we encode the data matrix $\bX$ using an encoding matrix $\bR=[\bR_1\ \bR_2\ \hdots\ \bR_m]$, where each $\bR_i$ is a $d\times p$ matrix (with $pm\geq d$), and store $\bX\bR_i$ at the $i$'th worker. Define $\btX^R:=\bX\bR$. 
Now, instead of solving \eqref{eq:problem-cd}, we solve the encoded problem $\arg\min_{\bv\in\R^{pm}}\phi(\btX^R\bv;\by)$ using Algorithm 1 (together with decoding at the master); see \Figureref{problem-setup_CD} on \Pageref{fig:problem-setup_CD} for a pictorial description of our algorithm. 
We design the encoding matrix $\bR$ such that at every iteration of our algorithm, updating any (small) subset of coordinates of $\bv_i$'s (let $\bv=[\bv_1^T\ \bv_2^T\ \hdots\ \bv_m^T]$) automatically updates some (small) subset of coordinates of $\bw$; 
and, furthermore, by updating those coordinates of $\bv_i$'s, we can efficiently recover the correspondingly updated coordinates of $\bw$, despite the errors injected by the adversary.
In fact, at any iteration $t$, the encoded parameter vector $\bv_t$ and the original parameter vector $\bw_t$ satisfies $\bv_t = \bR^+\bw_t$, where $\bR^+:=\bR^T(\bR\bR^T)^{-1}$ is the Moore-Penrose pseudo-inverse of $\bR$, and $\bw_t$ evolves in the same way as if we are running Algorithm 1 on the original problem.

We will be effectively updating the coordinates of the parameter vector $\bw$ in chunks of size $(m-2t)$ or its integer multiples (where $t$ is the number of corrupt workers). 
In particular, if each worker $i$ updates $k$ coordinates of $\bv_i$, then $k(m-2t)$ coordinates of $\bw$ will get updated.
For comparison, \Algorithmref{distr-coor-desc-algo} updates $km$ coordinates of the parameter vector $\bw$ in each iteration, if each worker updates $k$ coordinates in that iteration.

As described in \Algorithmref{distr-coor-desc-algo} for the Byzantine-free CD, in order to update its local parameter vector $\bw_i$ according to \eqref{eq:dist-coor-dist-update}, worker $i$ needs access to $\phi'(\bX\bw;\by)$, which master computes after receiving $\{\bX_j\bw_j\}_{j\in[m]}$ from the workers. 
In our Byzantine-resilient algorithm for CD also master will need to compute $\bX\bw$ in every CD iteration, and for this purpose, we employ the same encoding-decoding procedure for MV multiplication that we used in the first round of gradient computation, as described in \Subsectionref{approach_GD}. 
In particular, to make the notation distinct from gradient computation, in order to compute $\bX\bw$, we encode $\bX$ using an encoding matrix $\bL=[\bL_1^T\ \bL_2^T\ \hdots\ \bL_m^T]^T$, where each $\bL_i$ is a $p'\times n$ matrix (with $p'm\geq n$) and worker $i$ stores $\btX_i^L=\bL_i\bX$. 

Note that in order to compute $\bX\bw$, in the first round of gradient computation as described in \Subsectionref{approach_GD}, master broadcasts $\bw$ to all the workers and each worker $i$ computes $\btX_i^L\bw$ and sends it the the master (corrupt workers may report arbitrary vectors), who then decodes and obtains $\bX\bw$. 
However, in coordinate descent, though master wants to compute $\bX\bw$ in each CD iteration, we can significantly improve the computation required at each worker: 
since only a few coordinates of the original parameter vector $\bw$ are updated in each CD iteration, master needs to send only those updated coordinates, and workers need to
preform MV multiplication with a much smaller matrix, whose number of columns is equal to the number of updated coordinates of $\bw$ that they receive from master. Thus, the computational complexity in each CD iteration at worker is proportional to the number of coordinates updated in each CD iteration, as desired.

Our main result for the Byzantine-resilient distributed coordinate descent is stated below, which is proved in \Sectionref{solution_CD}.
\begin{thm}[Coordinate Descent]\label{thm:main-result_CD}
Under the setting of \Theoremref{main-result},
our Byzantine-resilient distributed CD algorithm has the following guarantees, where $\epsilon>0$ is a free parameter.
\begin{itemize}
\item $(s+t)\leq\left\lfloor\frac{\epsilon}{1+\epsilon}\cdot \frac{m}{2}\right\rfloor$.
\item Total storage requirement is roughly $2(1+\epsilon)|\bX|$.
\item If each worker $i$ updates $\tau$ coordinates of $\bv_i$, then 
\begin{itemize}
\item $\frac{\tau m}{1+\epsilon}$ coordinates of the corresponding $\bw$ gets updated.
\item the computational complexity in each iteration 
\begin{itemize}
\item at each worker node is $O(n\tau)$.
\item at the master node is $O((1+\epsilon)nm+\tau m^2)$.
\end{itemize}
\item the communication complexity in each iteration
\begin{itemize}
\item each worker sends $\left(\tau+(1+\eps)\frac{n}{m}\right)$ real numbers.
\item master broadcasts $\left(\frac{\tau m}{1+\eps}+n\right)$ real numbers.
\end{itemize}
\end{itemize}

\item Total encoding time is $O\left(nd\left(\frac{\epsilon}{1+\epsilon}m+1\right)\right)$. 
\end{itemize}
\end{thm}
\begin{remark}[Comparison with the plain distributed CD]\label{remark:comparison-with-dist-CD}
We compare the resource requirements of our method with the plain distributed CD described in \Algorithmref{distr-coor-desc-algo} that does not provide any adversarial protection.
Let $\eps$ be any number in the interval $[0,m-1]$ -- for illustration, we can take $\eps=2$, which means $t\leq\frac{m}{3}$ workers are corrupt.
In \Algorithmref{distr-coor-desc-algo}, if each worker $i$ updates $\frac{\tau}{1+\eps}$ coordinates of $\bw_i$ (in total $\frac{\tau m}{1+\eps}$ coordinates of $\bw$) in each iteration, then 
{\sf (i)} each worker requires $O(\frac{n\tau}{1+\eps})$ time to multiply $\bX_i$ with the updated part of $\bw_i$; 
{\sf (ii)} master requires $O(nm)$ time to compute $\sum_{i=1}^m\bX_i\bw_i$ from $\{\bX_i\bw_i\}_{i\in[m]}$; 
{\sf (iii)} each worker sends $n$ real numbers (required for $\bX_i\bw_i$) to master; and 
{\sf (iv)} master broadcasts $n$ real numbers (required for $\phi'(\bX\bw;\by)$).

In our scheme {\sf (i)} the total storage requirement is $O(1+\epsilon)$ factor more; {\sf (ii)} the amount of computation at each worker node is $O(1+\epsilon)$ factor more; {\sf (iii)} the amount of computation at the master node is $O((1+\epsilon)+\frac{\tau m}{n})$ factor more -- typically, since $\tau$ is a constant and number of workers is much less than $n$, this again could be $O(1+\eps)$; {\sf (iv)} master broadcasts $\left(1+\frac{\tau m}{(1+\eps)n}\right)$ factor more data, which could be a constant if $\tau m$ is smaller than $(1+\eps)n$; and {\sf (v)} each worker sends $\left(\frac{\tau}{n}+\frac{(1+\eps)}{nm}\right)$ factor more data, where the 1st term is much smaller than 1 as $\tau$ is typically a constant, and the 2nd term is close to zero as $(1+\eps)$ is always upper-bounded by $m$.
\end{remark}

\begin{remark}[Comparison with the replication-based strategy]\label{remark:comparison-replication_CD}
One simple way to make \Algorithmref{distr-coor-desc-algo} Byzantine-resilient is using repetition code, where we first divide the set of $m$ workers into $\frac{m}{2t+1}$ groups of size $(2t+1)$ each and also divide the data matrix as $\bX=[\bX_1\ \bX_2\ \hdots\ \bX_{\frac{m}{2t+1}}]$ (assume, for simplicity, that $(2t+1)$ divides $m$). Now, store the $i$'th block $\bX_i$ at the $(2t+1)$ workers in the $i$'th group of workers.
Let the parameter vector be divided as $\bw=[\bw_1^T\ \bw_2^T\ \hdots\ \bw_{\frac{m}{2t+1}}^T]^T$.
In each CD iteration, the local parameter updates in any $\bw_i$ is replicated at $(2t+1)$ different workers in the $i$'th group of workers, and since at most $t$ workers are corrupt, master can do a majority vote for decoding. Note that the total storage and the computation at workers in this scheme grow {\em linearly} by a factor of $(2t+1)$, where $t$ is the number of corruption, which could be significant. In contrast, the method that we propose can tolerate linear corruption, say, $t=\frac{m}{3}$, with a {\em constant} overhead in storage and computational complexity.
\end{remark}

The Remarks \ref{remark:errors-erasures}, \ref{remark:parameters}, \ref{remark:optimal-threshold} are also applicable for \Theoremref{main-result_CD}.

%% file: related-work.tex
\section{Related Work}\label{sec:related-work}
There has been a significant recent interest in using coding-theoretic
techniques to mitigate the well-known straggler problem
\cite{DeanTailAtScale}, including gradient coding
\cite{TandonLeDiKa17,RavivTaDiTa18,CharlesPa18,HalbawiRuSaHa18},
encoding computation \cite{LeeLaPePaRa18,ShortDotGroverNIPS16,ShortDotGrover_Journal19}, and data
encoding \cite{KarakusSuDiYi17,KarakusSuDiYi_JMLR19}. However, one cannot directly apply
the methods for straggler mitigation to the Byzantine attacks case, as
we do not know which updates are under attack. Distributed computing
with Byzantine adversaries is a richly investigated topic since
\cite{LamportShPe82}, and has received recent attention in the context
of large-scale distributed optimization and learning
\cite{BlanchardMhGuSt17,ChenSuXu17,ChenWaChPa18,YinChRaBa18,Alistarh_Byz-SGD18,SuX_Byz19,Zeno_ByzSGD19,Bartlett-Byz_nonconvex19,NirupamVa_Byz-SGD19,Detox_ByzSGD19,RSA_Byz-SGD19,YinRobustFL19,LargrangeCoding_Yu-etal19,DataDi_Byz-SGD_Heterogeneous20,DataDi_Byz-LocalSGD_Heterogeneous20,Jaggi-Heterogeneous20}.
These can be divided into three categories:
{\sf (i)} One which assume explicit statistical models for data across workers (e.g., data drawn i.i.d.\ from a probability distribution) and analyze gradient descent \cite{ChenSuXu17,YinChRaBa18,SuX_Byz19,Bartlett-Byz_nonconvex19,YinRobustFL19}.
{\sf (ii)} Other set of works make no probabilistic assumption on data, and optimize through stochastic methods (e.g., stochastic gradient descent) \cite{BlanchardMhGuSt17,Alistarh_Byz-SGD18,NirupamVa_Byz-SGD19,Zeno_ByzSGD19,RSA_Byz-SGD19,Detox_ByzSGD19,DataDi_Byz-LocalSGD_Heterogeneous20,DataDi_Byz-SGD_Heterogeneous20,Jaggi-Heterogeneous20} and also with deterministic methods (e.g., gradient descent) \cite{DataDi_Byz-LocalSGD_Heterogeneous20,DataDi_Byz-SGD_Heterogeneous20}.
Note that none of these two sets of works do data encoding and work with data as it is, 
and provide Byzantine resilience by applying some robust aggregation procedures (e.g., geometric median, coordinate-wise median, outlier-filtering, etc.) at the master for aggregating gradients. 
{\sf (iii)} Another line of work which is most relevant to ours provide Byzantine resiliency using redundant computations, either by encoding the gradients \cite{ChenWaChPa18} or by encoding the data itself \cite{LargrangeCoding_Yu-etal19}. 
Note that \cite{Detox_ByzSGD19} combines both redundant computations and do a hierarchical robust aggregation and not is directly comparable to ours. 

Note that the statistical nature of data/analysis in the first two sets of works leads to a statistical approximation error in the convergence rates, which is also intensified by the inaccuracy of the robust gradient aggregation procedure. 
One of the main focuses in these works is typically on obtaining faster convergence (where the goal is to match the convergence rate of plain SGD/GD) and as good an approximation error as possible. Note that the approximation error in all these works scales at least as $\Omega(\sqrt{d})$, where $d$ is the dimension of the model parameter vector, which may be significant in high-dimensional settings.
Moreover, in all these works, since we are not allowed to pre-process the data (such as, doing data encoding, etc.), we need to make some assumptions on the data, and furthermore, master has to apply a non-trivial decoding for gradient aggregation, which requires significantly more time than what our decoding requires.
For example, filtering-based decoding \cite{SuX_Byz19,DataDi_Byz-LocalSGD_Heterogeneous20,DataDi_Byz-SGD_Heterogeneous20}, median-based decoding \cite{ChenSuXu17,YinChRaBa18}, and heuristic approaches \cite{BlanchardMhGuSt17}, all have a {\em super-linear} complexity in $m$ -- in fact, the filtering-based method as in \cite{SuX_Byz19,DataDi_Byz-LocalSGD_Heterogeneous20,DataDi_Byz-SGD_Heterogeneous20} (which is the most effective in terms of the approximation error) requires $O(m^3d)$ time. In contrast, our decoding has a {\em linear} dependence on both $m$ and $d$. 
Note that, unlike the first two categories, the third line of work (to which ours also belongs) gives deterministic guarantees and work with arbitrary datasets, with no probabilistic assumptions; we elaborate on these and do a detailed comparison with ours below. We skip the comparison with the first two categories, as it would not be a fair comparison because the underlying setting is different -- results in the first two categories are based on statistical assumptions on data/algorithm and inaccurate gradient recovery, whereas, results in the third category make no assumption on the data/algorithm and allow exact gradient recovery.

We want to emphasize that all these works use gradient descent (GD) or stochastic gradient descent (SGD) as their optimization algorithm, which is a data-parallelization method; in this paper, additionally, we also use coordinate descent (CD) algorithm for optimization, which is a model-parallelization method and is preferred over GD in some applications; see \Sectionref{intro} for more details on this. As will be evident from \Sectionref{solution_CD}, making CD secure against Byzantine attacks is arguably more intricate than securing GD.

We divide this section into three categories: first we compare the redundancy-based methods for GD in \Subsectionref{related-work_GD}, and then CD in \Subsectionref{related-work_CD}. Since we use matrix-vector (MV) multiplication as a core subroutine for both GD and CD, we also compare related work on this in \Subsectionref{related-work_MV-mult}.

\subsection{Gradient Descent (GD)}\label{subsec:related-work_GD}
In this section, we do a detailed comparison with \cite{ChenWaChPa18} and \cite{LargrangeCoding_Yu-etal19}, which are the closest related works that also combat Byzantine adversaries using redundant computations.

For the sake of comparison, assume that $t\leq\frac{m-1}{2}$ workers are corrupt.
The coding scheme of Chen et al.~\cite{ChenWaChPa18}, which they called {\sc Draco},
requires repetition of each data point $(2t+1)$ times, 
storing each copy at different workers. This gives the storage redundancy {\em factor} of $(2t+1)$ in {\sc Draco}, whereas, our coding method requires storage redundancy factor of $2(1+\eps)=\frac{2m}{m-2t}$, which is a constant \emph{even} if $t$ is a constant ($<\frac{1}{2}$) fraction of $m$.\footnote{To highlight the storage redundancy gain of our method over that of {\sc Draco}, consider the following two concrete scenarios, where the data matrix $\bX\in\R^{n\times d}$ consists of $nd$ real numbers: {\sf (i)} In a large setup with $m=1000$ worker nodes, if we want resiliency against $t=100$ corrupt nodes (1/10 nodes are corrupt), our method requires redundancy of 2.5, whereas {\sc Draco} requires redundancy of 201 (i.e., we need to store only $2.5\times nd$ real numbers, whereas {\sc Draco} stores $201\times nd$ real numbers), a multiplicative-factor of $>$ 80 more than ours. {\sf (ii)} In a moderate setup with $m=150$ and $t=50$ (1/3 nodes are corrupt), the redundancy of our method is 6, whereas {\sc Draco} requires redundancy of 101, a multiplicative-factor of $\approx 17$ more than ours.} 
Since each worker in {\sc Draco} is doing $(2t+1)$-factor more computation for each GD iteration (than simply computing the gradients as in plain distributed GD), the computational cost at workers also grows by the same factor, which is a significant downside of their scheme. In contrast, our scheme only requires $O(\frac{m}{m-2t})$ more computation at worker, which is a constant even if $t$ is a constant ($<\frac{1}{2}$) fraction of $m$.
This significantly reduces the computation time at the worker nodes in our
scheme compared to {\sc Draco}, without
sacrificing much on the computation time required by the master node -- 
the decoding at master in {\sc Draco} takes $O(md)$ time, whereas, our scheme requires $O(\frac{m}{m-2t}(n+d)m)$ time, which is a factor of $O(\frac{m}{m-2t}(1+\frac{n}{d}))$ more than {\sc Draco}. In high-dimensional settings, where $n$ is not much bigger than $d$, and $t$ is a constant ($<\frac{1}{2}$) fraction of $m$, this overhead is constant. Overall, for a constant fraction of corruption, say, $t=\frac{m}{3}$, {\sc Draco} requires $\Omega(t)$ times more storage and computation at workers than our scheme (which could be significant in large-scale settings), and requires $\Omega(1+\frac{n}{d})$ times less computation at master. 
Note that the computation time at workers scales at least as $\Omega(\frac{nd}{m})$, which dominates the time taken by master (since $n,d$ are typically much larger than $m$), so our scheme will be faster than {\sc Draco} with respect to the overall running time. 
Note that the coding in {\sc Draco} is restricted to data replication
redundancy, as they encode the gradient as done in
\cite{TandonLeDiKa17}, enabling application to (non)-convex problems; in
contrast, we encode the data enabling significantly smaller
redundancy, and apply it to learn generalized linear models, and is also
applicable to MV multiplication.

Yu et al.~\cite{LargrangeCoding_Yu-etal19} (which is a concurrent work\footnote{\label{foot:lagrange}Yu et al.~\cite{LargrangeCoding_Yu-etal19} is concurrent to our conference versions in Allerton 2018 \cite{DataSoDi-ByzGD-Allerton18} and ISIT 2019 \cite{DataSoDi-ByzGD-ISIT19,DataDi_ByzCD-ISIT19}, on which this paper is based.}) proposes Lagrange coded computing in a distributed framework to compute any multivariate polynomial of the input data and simultaneously provides resilience against stragglers, security against adversaries, and privacy of the dataset against collusion of workers. They leverage the Lagrange polynomial to create computation redundancy among workers, and using standard Reed-Solomon decoding, they can tolerate both erasures/stragglers and errors/adversaries.
Their method provide privacy by adding random elements from the field (which in the case of gradient computation is the field of all matrices of a certain dimension) while doing the polynomial interpolation. This is a standard method in Shamir secret sharing scheme \cite{ShamirSSS79} that is widely used in information-theoretically secure MPC protocols \cite{MPCBook15} to provide privacy of users' data.
For the sake of comparison of the resource requirements of our scheme and the one in \cite{LargrangeCoding_Yu-etal19}, consider the task of linear regression (the concrete machine learning application studied in \cite{LargrangeCoding_Yu-etal19}). In the following, we assume that $\frac{m-1}{2}-\delta$ workers are corrupt, which corresponds to $\eps=\frac{m}{1+2\delta}-1$ in our setting; here $\delta$ can take any value in $[0:\frac{m-1}{2}]$. {\sf (i)} The storage overhead of our scheme is $\frac{m}{\delta+\nicefrac{1}{2}}$, whereas, in \cite{LargrangeCoding_Yu-etal19}, it is $\frac{m}{\delta+1}$, which is roughly the same as ours.
For example, to tolerate $\frac{m}{3}$ corrupt workers (i.e., $\delta=\frac{m-3}{6}$), the storage overhead of our scheme and of \cite{LargrangeCoding_Yu-etal19} is a multiplicative factor of $6$ and $\frac{6}{1+\nicefrac{3}{m}}\approx 6$, respectively.
{\sf (ii)} The encoding time complexity of our scheme is $O(nd(m-2\delta))$, whereas, it is $O(m\log^2(m)\frac{nd}{\delta+1})$ in \cite{LargrangeCoding_Yu-etal19}. 
Note that for constant $\delta$ (i.e., corruption close to $1/2$), the encoding time of our scheme is much less (by a factor of $O(m\log^2(m))$) than that of \cite{LargrangeCoding_Yu-etal19}, whereas, for corruption $cm$, where $c<\frac{1}{2}$, the scheme of \cite{LargrangeCoding_Yu-etal19} takes $O(\frac{m}{\log^2(m)})$-factor less time in encoding than ours.
{\sf (iii)} The computation time at each worker per gradient computation in both our scheme and \cite{LargrangeCoding_Yu-etal19} is roughly the same -- ours requires $O(\frac{nd}{1+2\delta})$ time and \cite{LargrangeCoding_Yu-etal19} requires $O(\frac{nd}{1+\delta})$ time.
{\sf (iv)} The decoding time complexity per gradient computation in \cite{LargrangeCoding_Yu-etal19} is $O(m\log^2(m)d)$, whereas, ours requires $O((1+\eps)(n+d)m)$ time. Note that when $n$ is not much bigger than $d$ and we want a constant fraction of corruption, say, $\frac{m}{3}$ corruption, then their decoding complexity is worse than ours by a logarithmic factor.
Also note that our decoding algorithm is arguably simpler than theirs.
{\sf (v)} For per gradient computation, each worker respectively sends $\frac{n+d}{1+2\delta}$ and $d$ real numbers in ours and the scheme in \cite{LargrangeCoding_Yu-etal19}. Note that if $n\leq dm$ and to tolerate a constant fraction of corruption, say, $\frac{m}{3}$ corruption, each worker sends roughly $O(m)$ less data in our scheme than that of \cite{LargrangeCoding_Yu-etal19}. Overall, if we want tolerance against $\frac{m}{3}$ corrupt worker nodes, then both our scheme and the one in \cite{LargrangeCoding_Yu-etal19} have similar resource requirements, except for that our scheme has a much better communication complexity (by a factor of $O(m)$) from workers to the master, whereas, the encoding time complexity (which is a one-time process) of \cite{LargrangeCoding_Yu-etal19} is better than ours by a factor of $O(\frac{m}{\log^2(m)})$. 

\subsection{Coordinate Descent (CD)}\label{subsec:related-work_CD} 
Even for the straggler problem, we are only aware of one work by Karakus et al.~\cite{KarakusSuDiYi_JMLR19} that, in addition to distributed GD, also studies distributed CD, and that for quadratic problems (e.g., linear/ridge regression) only.
It also does data encoding and achieves low redundancy and low complexity, by allowing convergence to an approximate rather than exact solution. 
As far as we know, ours is the first work that studies distributed CD under Byzantine attacks and provides an efficient solution, much better than the replication-based solution (see \Remarkref{comparison-replication_CD}). 
At the heart of our solution for CD is the matrix-vector (MV) multiplication procedure that we develop in this paper; and it is the specific regular structure of our encoding matrix (given in \eqref{eq:encoding-matrix-S_i}, designed for the MV multiplication) that allows for partially updating the coordinates of the parameter vector in each CD iteration. 
Note that a general-purpose encoding matrix for MV multiplication may not be applicable for the CD algorithm.

It has been observed earlier in several works (see, for example, \cite{LeeLaPePaRa18,ShortDotGroverNIPS16}) that gradient computation in GD for linear regression can be reduced to MV multiplication, and any general-purpose code for MV multiplication can be used to provide a solution for gradient computation. 
As far as we know, ours is the first paper that makes the connection of CD and MV multiplication, and provides an efficient solution for CD (which is also resilient to Byzantine attacks) for learning generalized linear models. Note that, unlike GD, not any general-purpose code for MV multiplication can be used for CD: the main challenge in CD comes from the fact that we only update a small number of coordinates of the parameter vector in each CD iteration; when we encode the data and iteratively update some coordinates of the (encoded) parameter vector using the encoded data, we need to make sure that this update in the encoded parameter vector is reconciled with the update in the original parameter vector. This is fundamentally different from GD iterations. See \Sectionref{solution_CD} for more details.

\subsection{Matrix-Vector Multiplication}\label{subsec:related-work_MV-mult}
For the task of a more fundamental problem of matrix-vector (MV) multiplication in the presence of Byzantine adversaries, which is at the core of the optimization algorithms in this paper, we are only aware of two concurrent works \cite{LargrangeCoding_Yu-etal19} (see \Footnoteref{lagrange}) and \cite{ShortDotGrover_Journal19}\footnote{\label{foot:shortdot}The conference version \cite{ShortDotGroverNIPS16} only studies the straggler problem, and the journal version \cite{ShortDotGrover_Journal19} briefly mentions how their results from \cite{ShortDotGroverNIPS16} can be extended to handle adversarial nodes, and we describe that in this section.} that provide (coding-theoretic) solutions to this problem. In the following, we do a detailed comparison of our solution with both of these works and also discuss the (dis)similarities.

We have already done a detailed comparison with Yu et al.~\cite{LargrangeCoding_Yu-etal19} (concurrent work, see \Footnoteref{lagrange}) with respect to gradient descent in \Subsectionref{related-work_GD}. For the problem of MV multiplication, the storage requirement, computation time per worker, and communication complexity to/from workers is the same in both ours and \cite{LargrangeCoding_Yu-etal19}. The comparison of encoding time complexity is same as above; however, for a constant corruption, say, $\frac{m}{3}$ corrupt workers, our method outperforms the one in \cite{LargrangeCoding_Yu-etal19} in terms of the decoding time complexity by a factor of $O(\log^2(m))$. Note that, unlike \cite{LargrangeCoding_Yu-etal19}, we make a fundamental connection of handling Byzantine errors with the sparse reconstruction (or the real-error correction) problem from the compressive sensing literature \cite{CandesTao05}.

Dutta et al.~\cite{ShortDotGrover_Journal19} (concurrent work, see \Footnoteref{shortdot}) focuses on matrix-vector (MV) multiplication. Though their main focus is on providing resilience against stragglers, they also mention that handling stragglers is very different than handling errors, as it requires to correct errors over real numbers, and, unlike stragglers, we do not know which workers are corrupt. Similar to our observation, they also note that since the matrices and vectors have entries from real numbers, the decoding problem reduces to the sparse reconstruction problem from the compressive sensing literature \cite{CandesTao05} and they also provide such a reduction.
Apart from these similarities, our solution for MV multiplication differs from that of \cite{ShortDotGrover_Journal19} in several important ways: 
{\sf (i)} \cite{ShortDotGrover_Journal19} provides a detailed solution to the distributed MV multiplication for the straggler problem for the case when the number of rows in the matrix is smaller than the number of workers nodes. As mentioned in \cite{ShortDotGrover_Journal19}, this method can be easily generalized to the more general case when the matrix is of arbitrary dimension, in which case, first we can divide the rows of the matrix into several sub-matrices, each having number of rows smaller than the number of workers, and then apply the above method independently to each sub-matrix. This simple extension may work (without losing efficiency) for the straggler/erasure problem, however, leads to a highly inefficient solution for the adversary/error problem. 
The reason being that, in the presence of Byzantine workers, if we solve the sparse reconstruction problem for each sub-matrix separately, this would be inefficient, as the decoding would then be computationally expensive. 
To remedy this, we exploit the block error pattern and use a simple idea of linearly combining the response vectors from each worker using coefficients drawn from an absolutely continuous distribution, so that we only need to do {\em just one} computation for solving the sparse construction problem. This significantly reduces the decoding complexity; see \Subsectionref{finding-corrupt-nodes} for details.
{\sf (ii)} \cite{ShortDotGrover_Journal19} only shows a connection to the sparse recovery problem, whereas, we provide a complete solution, with a concrete sparse recovery (or real-error correction) matrix and resource (encoding/decoding time, storage, communication) requirement analysis. 
{\sf (iii)} Our encoding matrix (given in \eqref{eq:encoding-matrix-S_i}) to encode data matrices of arbitrary dimensions is very sparse and highly structured which allows us to apply that construction to CD algorithm, which, as far we know, has not been connected with MV multiplication before. Also, ours is the first paper that provides a non-trivial and efficient (data encoding) solution to CD in the presence of a Byzantine adversary.
{\sf (iv)} We also want to mention that the focus in \cite{ShortDotGrover_Journal19} is on making the {\em encoded} matrix sparse (at the expense of increased computation at workers) so that workers need to compute shorter dot products, whereas, in this paper, we make the {\em encoding} matrix sparse (much sparser than the {\em encoded} matrix of \cite{ShortDotGrover_Journal19}) to get efficient encoding/decoding.

%% file: solution_GD.tex
\section{Our Solution to Gradient Computation}\label{sec:matrix-vector-mult}
In this section, we describe the core technical part of our two-round approach for gradient computation described in \Subsectionref{approach_GD} -- a method for performing matrix-vector (MV) multiplication in a distributed manner in the presence of a malicious adversary who can corrupt at most $t$ of the $m$ worker nodes. Here, the matrix is fixed and we want to right-multiply a vector with this matrix.

Given a fixed matrix $\bA\in\R^{\nr\times \nc}$ and a vector $\bv\in\R^{\nc}$, we want to compute $\bA\bv$ in a distributed manner in the presence of at most $t$ corrupt worker nodes; see \Subsectionref{adversary-model} for details on our adversary model.
Our method is based on data encoding and error correction over real numbers, 
where the matrix $\bA$ is encoded and distributed among all the worker nodes, 
and the master node recovers the MV product $\bA\bv$ using real-error correction; see \Figureref{problem-setup}.
We will think of our encoding matrix as 
$\bS=[\bS_1^T\ \bS_2^T,\hdots,\bS_m^T]$, where each $\bS_i$ is a $p\times\nr$ matrix and $pm\geq \nr$.
We will derive the matrix $\bS$ in \Subsectionref{encoding-matrix}.
For the value of $p$, looking ahead, we will set $p=\lceil\frac{n}{m-2t}\rceil$, which is a constant multiple of $\frac{n}{m}$ even if $t$ is a constant ($<\frac{1}{2}$) fraction of $m$ (e.g., if $t=\frac{m}{3}$, we would have $p=\frac{3n}{m}$).
For $i\in[m]$, we store the matrix $\bS_i\bA$ at the $i$'th worker node.
As described in \Sectionref{prob-and-result}, the computation proceeds as follows:
The master sends $\bv$ to all the worker nodes and receives $\{\bS_i\bA\bv + \be_i\}_{i=1}^m$ back from them.
Let $\be_i=[e_{i1},e_{i2},\hdots,e_{ip}]^T$ for every $i\in[p]$.
Note that $\be_i=\bzero$ if the $i$'th node is honest, otherwise can be arbitrary. 
In order to find the set of corrupt worker nodes, 
master equivalently writes $\{\bS_i\bA\bv+\be_i\}_{i=1}^m$ as $p$ systems of linear equations.
\begin{equation}\label{eq:subsystems}
\tilde{h}_i(\bv) = \btS_i\bA\bv + \bte_i,\quad i\in[p]
\end{equation}
where, for every $i\in[p]$, $\bte_i=[e_{1i},e_{2i},\hdots,e_{mi}]^T$, and 
$\btS_i$ is an $m\times\nr$ matrix whose $j$'th row is equal to the $i$'th row of $\bS_j$, for every $j\in[m]$. 
Note that at most $t$ entries in each $\bte_i$ are non-zero. 
Observe that $\{\bS_i\bA\bv+\be_i\}_{i=1}^m$ and $\{\btS_i\bA\bv+\bte_i\}_{i=1}^p$ 
are equivalent systems of linear equations, and we can get one from the other.
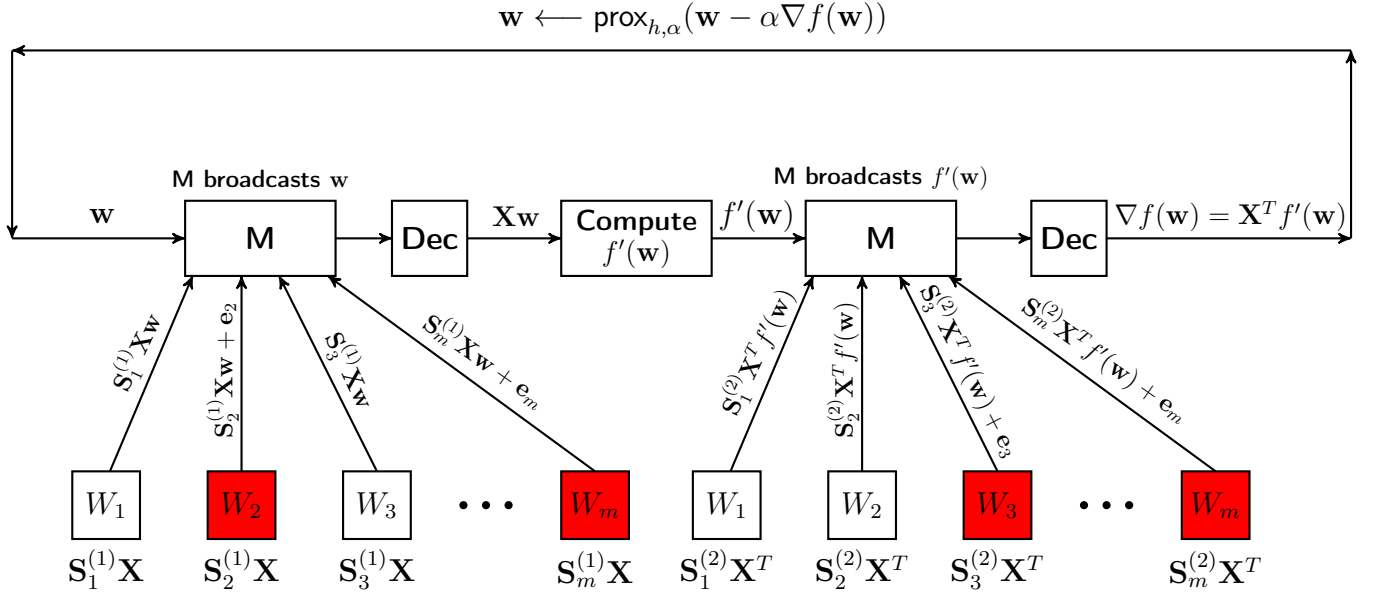
\begin{figure*}[t]
\centering
\begin{tikzpicture}[>=stealth', font=\sffamily\large\bfseries, thick]
\draw [->] (-0.8,4.0) -- (1.5,4.0); \node at (0.4,4.3) {$\bw$};
\node [rotate=0, scale=0.75] at (2.5,4.8) {M \text{broadcasts} $\bw$};
\draw (1.5,3.5) rectangle (3.5,4.5) node [pos=0.5] {M};
\draw [->] (3.5,4.0) -- (4.25,4.0);
\draw (4.25,3.5) rectangle (5.25,4.5) node [pos=0.5] {Dec};

\draw (0,0) rectangle (0.9,0.9) node [pos=0.5] {$W_1$}; \node at (0.45, -0.4) {$\bS_1^{(1)}\bX$};
\draw [fill=red] (1.8,0) rectangle (2.7,0.9) node [pos=0.5] {$W_2$}; \node at (2.25, -0.4) {$\bS_2^{(1)}\bX$};
\draw (3.6,0) rectangle (4.5,0.9) node [pos=0.5] {$W_3$}; \node at (4.05, -0.4) {$\bS_3^{(1)}\bX$};
\draw [fill=black] (5.2,0.45) circle [radius=0.04];
\draw [fill=black] (5.5,0.45) circle [radius=0.04];
\draw [fill=black] (5.8,0.45) circle [radius=0.04];
\draw [fill=red] (6.5,0) rectangle (7.4,0.9) node [pos=0.5] {$W_m$}; \node at (6.95, -0.4) {$\bS_m^{(1)}\bX$};

\draw [->] (0.5,0.9) -- (1.6,3.5); \node [rotate=65, scale=0.8] at (0.85,2.5) {$\bS_1^{(1)}\bX\bw$};
\draw [->] (2.25,0.9) -- (2.25,3.5); \node [rotate=90, scale=0.8] at (2.0,2.25) {$\bS_2^{(1)}\bX\bw+\be_2$};
\draw [->] (4.05,0.9) -- (2.75,3.5); \node [rotate=-65, scale=0.8] at (3.7,2.25) {$\bS_3^{(1)}\bX\bw$};
\draw [->] (6.95,0.9) -- (3.4,3.5); \node [rotate=-35, scale=0.8] at (5.5,2.4) {$\bS_m^{(1)}\bX\bw+\be_m$};

\draw [->] (5.25,4.0) -- (6.5,4.0); 
\node [scale=0.9] at (5.9,4.3) {$\bX\bw$};
\draw (6.5,3.5) rectangle (8.5,4.5); 
\node [scale=0.9] at (7.5,4.2) {Compute};
\node [scale=0.9] at (7.5,3.8) {$f'(\bw)$};
\draw [->] (8.5,4.0) -- (9.75,4.0); 
\node at (9.1,4.3) {$f'(\bw)$};

\node [rotate=0, scale=0.75] at (10.75,4.8) {M \text{broadcasts} $f'(\bw)$};
\draw (9.75,3.5) rectangle (11.75,4.5) node [pos=0.5] {M};
\draw [->] (11.75,4.0) -- (12.75,4.0);
\draw (12.75,3.5) rectangle (13.75,4.5) node [pos=0.5] {Dec};

\draw (8.25,0) rectangle (9.15,0.9) node [pos=0.5] {$W_1$}; \node at (8.7, -0.4) {$\bS_1^{(2)}\bX^T$};
\draw (10.05,0) rectangle (10.95,0.9) node [pos=0.5] {$W_2$}; \node at (10.5, -0.4) {$\bS_2^{(2)}\bX^T$};
\draw [fill=red] (11.85,0) rectangle (12.75,0.9) node [pos=0.5] {$W_3$}; \node at (12.30, -0.4) {$\bS_3^{(2)}\bX^T$};
\draw [fill=black] (13.45,0.45) circle [radius=0.04];
\draw [fill=black] (13.75,0.45) circle [radius=0.04];
\draw [fill=black] (14.05,0.45) circle [radius=0.04];
\draw [fill=red] (14.75,0) rectangle (15.65,0.9) node [pos=0.5] {$W_m$}; \node at (15.2, -0.4) {$\bS_m^{(2)}\bX^T$};

\draw [->] (8.75,0.9) -- (9.85,3.5); \node [rotate=65, scale=0.8] at (9.1,2.5) {$\bS_1^{(2)}\bX^Tf'(\bw)$};
\draw [->] (10.5,0.9) -- (10.5,3.5); \node [rotate=90, scale=0.8] at (10.25,2.25) {$\bS_2^{(2)}\bX^Tf'(\bw)$};
\draw [->] (12.3,0.9) -- (11.0,3.5); \node [rotate=-63, scale=0.8] at (11.95,2.25) {$\bS_3^{(2)}\bX^Tf'(\bw)+\be_3$};
\draw [->] (15.2,0.9) -- (11.65,3.5); \node [rotate=-35, scale=0.8] at (13.75,2.4) {$\bS_m^{(2)}\bX^Tf'(\bw)+\be_m$};

\draw [->] (13.75,4.0) -- (17.0,4.0); \node [scale=0.9] at (15.4,4.3) {$\nabla f(\bw)=\bX^Tf'(\bw)$};

\draw [->] (17.0,4.0) -- (17.0,6.5); \draw [->] (17.0,6.5) -- (-0.8,6.5); \draw [->] (-0.8,6.5) -- (-0.8,4.0);

\node at (8.25,6.9) {$\bw \longleftarrow {\sf prox}_{h,\alpha}(\bw - \alpha\nabla f(\bw))$};
\end{tikzpicture}
\caption{This figure shows our 2-round approach to the Byzantine-resilient distributed gradient descent to optimize \eqref{eq:glm} for learning a generalized linear model. 
Since the gradient at $\bw$ is equal to $\nabla f(\bw)=\bX^Tf'(\bw)$ (see \eqref{eq:gradient-2-step}), we compute it in 2 rounds, using a matrix-vector (MV) multiplication as a subroutine in each round. In the 1st round, first we compute $\bX\bw$, and then compute $f'(\bw)$ from $\bX\bw$ -- since the $j$'th entry of $\bX\bw$ is equal to $\langle\bx_j,\bw\rangle$, we can compute $f'(\bw)$ from $\bX\bw$ (see \Subsectionref{approach_GD}). In the 2nd round we compute $\bX^Tf'(\bw)$ -- which is equal to $\nabla f(\bw)$ -- using another application of MV multiplication.
For a matrix $\bA$ and a vector $\bv$, to make our distributed MV multiplication $\bA\bv$ Byzantine-resilient, we encode $\bA$ using a sparse matrix $\bS=[\bS_1^T\ \bS_m^T\ \hdots\ \bS_m^T]^T$ and distribute $\bS_i\bA$ to worker $i$ (denoted by $W_i$). 
Note that in the first round, we have $\bA=\bX, \bv=\bw$, and we encode $\bX$ using $\bS^{(1)}$, and in the second round, we have $\bA=\bX^T,\bv=f'(\bw)$, and encode $\bX^T$ using $\bS^{(2)}$.
The adversary can corrupt at most $t$ workers (the compromised ones are denoted in red color), potentially different sets of $t$ workers in different rounds.
The master node (denoted by {\bf M}) broadcasts $\bv$ to all the workers.
Each worker performs the local MV product and sends it back to M.
If $W_i$ is corrupt, then it can send an arbitrary vector. Once the master has received all the vectors (out of which $t$ may be erroneous), it sends them to the decoder (denoted by {\bf Dec}), which outputs the correct MV product $\bA\bv$.}
\label{fig:problem-setup}
\end{figure*}

Note that $\btS_i$'s constitute the encoding matrix $\bS$, which we have to design. In the following, we will design these matrices $\btS_i$'s (which in turn will determine the encoding matrix $\bS$), with the help of another matrix $\bF$, which will be used to find the error locations, i.e., identities of the compromised worker nodes. 
We will design the matrix $\bF$ (of dimension $k\times m$, where $k<m$ -- here $k$ is determined by the error-correction capability, and we will set $k=2t$; see \Subsectionref{F-matrix} for more details) and the matrices $\btS_i$'s such that 
\begin{enumerate}\renewcommand{\labelenumi}{\textbf{C.\theenumi}}
\item \label{cond1} $\bF\btS_i=0$ for every $i\in[p]$.
\item \label{cond2} For any $t$-sparse $\bu\in\R^m$, we can efficiently find all the non-zero locations of $\bu$ from $\bF\bu$.
\item \label{cond3} For any $\T\subset[m]$ such that $|\T|\geq (m-t)$, let $\bS_{\T}$ denote the $|\T|p\times \nr$ matrix obtained from $\bS$ by restricting it to all the $\bS_i$'s for which $i\in\T$. We want $\bS_{\T}$ to be of full column rank. 
\end{enumerate} 
If we can find such matrices, then we can recover the desired MV multiplication $\bA\bv$ exactly: briefly, 
{\bf C.\ref{cond1}} and {\bf C.\ref{cond2}} will allow us to locate the corrupt worker nodes; once we have found them, we can discard all
the information that the master node had received from them. 
This will yield $\bS_{\T}\bA\bv$, where $\bS_{\T}$ is the $|\T|p\times \nr$ matrix obtained from $\bS$ by restricting it to $\bS_i$'s for all $i\in\T$, where $\T$ is the set of all honest worker nodes.
Now, by {\bf C.\ref{cond3}}, since $\bS_{\T}$ is of full column rank, we can recover $\bA\bv$ from $\bS_{\T}\bA\bv$ exactly. Details follow.

Suppose we have matrices $\bF$ and $\btS_i$'s such that {\bf C.\ref{cond1}} holds. 
Now, multiplying \eqref{eq:subsystems} by $\bF$ yields
\begin{align}
\bbf_i := \bF\tilde{h}_i(\bv)=\bF\bte_i,  \label{eq:recover-error}
\end{align}
for every $i\in[p]$, where $\|\bte_i\|_0\leq t$. 
In \Subsectionref{finding-corrupt-nodes}, we give our approach for finding all the corrupt worker nodes with the help of any error locator matrix $\bF$. 
Then, in \Subsectionref{encoding-matrix}, we give a generic construction for designing $\btS_i$'s 
(and, in turn, our encoding matrix $\bS$)  such that {\bf C.\ref{cond1}} and {\bf C.\ref{cond3}} hold.
In \Subsectionref{recovery_matvecprod}, we show how to compute the desired matrix-vector product $\bA\bv$ efficiently, once we have discarded all the data from the corrupt works nodes. Then, in \Subsectionref{F-matrix}, we will give details of the error locator matrix $\bF$ that we use in our construction.

\begin{remark}\label{remark:struc_S-inde-F}
As we will see in \Subsectionref{encoding-matrix}, the structure of our encoding matrix $\bS$ is independent of our error locator matrix $\bF$. Specifically, the repetitive structure of the non-zero entries of $\bS$ as well as their locations will not change irrespective of what the $\bF$ matrix is. This makes our construction very generic, as we can choose whichever $\bF$ suits our needs the best (in terms of how many erroneous indices it can locate and with what decoding complexity), and it won't affect the structure of our encoding matrix at all -- only the non-zero entries might change, neither their repetitive format, nor their locations!
\end{remark}

\subsection{Finding The Corrupt Worker Nodes}\label{subsec:finding-corrupt-nodes}
Observe that $\supp(\bte_i)$ may not be the same for all $i\in[p]$, 
but we know, for sure, that the non-zero locations in all these error vectors occur within the same set of $t$ locations.
Let $\I = \bigcup_{i=1}^p\supp(\bte_i)$, which is the set of all corrupt worker nodes.
Note that $|\I|\leq t$.
We want to find this set $\I$ efficiently, and for that we note the following crucial observation.
Since the non-zero entries of all the error vectors $\bte_i$'s occur in the same set $\I$,
a random linear combination of $\bte_i$'s has support equal to $\I$ with probability one, if the coefficients of the linear combination are chosen from an {\em absolutely continuous} probability distribution.
This idea has appeared before in \cite{MishaliEl08} in the context of compressed sensing for recovering arbitrary sets of jointly sparse signals that have been measured by the same measurement matrix.
\begin{defn}\label{defn:abso-cont}
A probability distribution is called {\em absolutely continuous}, if every event of measure zero occurs with probability zero. 
\end{defn}
It is well-known that a distribution is absolutely continuous if and only if it can be represented as an integral over an integrable density function~\cite[Theorem 31.8, Chapter 6]{Billingsley95}. Since Gaussian and uniform distributions have an explicit integrable density function, both are absolutely continuous. Conversely, discrete distributions are not absolutely continuous.
Now we state a lemma from \cite{MishaliEl08} that shows that a random linear combination of the error vectors 
(where coefficients are chosen from an absolutely continuous distribution) preserves the support with probability one.
\begin{lem}[\hspace{-0.001cm}\cite{MishaliEl08}]\label{lem:random-combination}
Let {\em $\I = \bigcup_{i=1}^p\supp(\bte_i)$}, and
let $\bhe = \sum_{i=1}^p \alpha_i\bte_i$, where $\alpha_i$'s are sampled i.i.d.~from an absolutely continuous distribution. 
Then with probability 1, we have {\em $\supp(\bhe) = \I$}.
\end{lem}
From \eqref{eq:recover-error} we have $\bbf_i=\bF\bte_i$ for every $i\in[p]$. 
Take a random linear combination of $\bbf_i$'s with coefficients $\alpha_i$'s chosen i.i.d.~from an absolutely continuous distribution, for example, the Gaussian distribution. Let $\btf = \alpha_i\left(\sum_{i=1}^p\bbf_i\right) = \alpha_i\left(\sum_{i=1}^p\bF\bte_i\right) = \bF\left(\sum_{i=1}^p\alpha_i\bte_i\right) = \bF\bte$, where $\bte = \sum_{i=1}^p\alpha_i\bte_i$.
Note that, with probability 1, $\supp(\bte)$ is equal to the set of all corrupt worker nodes, 
and we want to find this set efficiently. 
In other words, given $\bF\bte$, we want to find $\supp(\bte)$ efficiently.
For this, we need to design a $k\times m$ matrix $\bF$ (where $k<m$) such that for any sparse error vector $\be\in\R^m$, we can efficiently find $\supp(\be)$ from $\bbf=\bF\be$. 
Many such matrices have been known in the literature that can handle different levels of sparsity with varying decoding complexity. We can choose any of these matrices depending on our need, and this will not affect the design of our encoding matrix $\bS$. 
In particular, we will use a $k\times m$ Vandermonde matrix along with the Reed-Solomon type decoding,
which can correct up to $k/2$ errors and has decoding complexity of $O(m^2)$; 
see \Subsectionref{F-matrix} for details.

\paragraph{Time required in finding the corrupt worker nodes.}
The time taken in finding the corrupt worker nodes is equal to the sum of the time taken in the following 3 tasks.
{\sf (i)} Computing $\bF\bte_i$ for every $i\in[p]$: Note that we can get $\bF\bte_i$ by multiplying \eqref{eq:subsystems} with $\bF$. Since  $\bF$ is a $k\times m$ matrix, and we compute $\bF\tilde{h}_i(\bv)$ for $p$ systems, this requires $O(pkm)$ time.
{\sf (ii)} Taking a random linear combination of $p$ vectors each of length $m$, which takes $O(pm)$ time.
{\sf (iii)} Applying \Lemmaref{RS-sparse-decoding} (in \Subsectionref{F-matrix}) once to find the error locations, which takes $O(m^2)$ time.
Since $p$ is much bigger than $m$, the total time complexity is $O(pkm)$.

\subsection{Designing The Encoding Matrix $\bS$}\label{subsec:encoding-matrix}
Now we give a generic construction for designing $\btS_i$'s such that {\bf C.\ref{cond1}} and {\bf C.\ref{cond3}} hold. 
Fix any $k\times m$ matrix $\bF$ such that we can efficiently find $\be$ from $\bF\be$, provided $\be$ is sufficiently sparse.
We can assume, without loss of generality, that $\bF$ has full row-rank; otherwise, there will be redundant observations in $\bF\be$ that we can discard and make $\bF$ smaller by discarding the redundant rows. 
Let $\N(\bF)\subset\R^m$ denote the null-space of $\bF$.
Since $\textsf{rank}(\bF)=k$, dimension of $\N(\bF)$ is $q=(m-k)$.
Let $\{\bb_1,\bb_2,\hdots,\bb_q\}$ be a basis of $\N(\bF)$, and let $\bb_i=[b_{i1}\ b_{i2}\hdots b_{im}]^T$, for every $i\in[q]$. 
We set $\bb_i$'s the columns of the following matrix $\bF^{\perp}$:
\begin{align}\label{eq:F-perp}
\bF^{\perp} = 
\begin{bmatrix}
b_{11} & b_{21}  & \hdots & b_{q1} \\
b_{12} & b_{22}  & \hdots & b_{q2} \\
\vdots & \vdots & \vdots & \vdots \\
b_{1m} & b_{2m}  & \hdots & b_{qm}
\end{bmatrix}_{m\times q}
\end{align}
The following property of $\bF^{\perp}$ will be used for recovering the MV product in \Subsectionref{recovery_matvecprod}.
\begin{claim}\label{claim:rest-F-perp_full-rank}
For any subset $\T\subset[m]$, such that $|\T|\geq (m-t)$, let $\bF_{\T}^{\perp}$ be the $|\T|\times q$ matrix, which is equal to the restriction of $\bF^{\perp}$ to the rows in $\T$. Then $\bF_{\T}^{\perp}$ is of full column rank.
\end{claim}
\begin{proof}
Note that $q=m-k$, where $k=2t$. So, if we show that any $q$ rows of $\bF^{\perp}$ are linearly independent, 
then, this in turn will imply that for every $\T\subset[m]$ with $|\T|\geq (m-t)$, the sub-matrix $\bF_{\T}^{\perp}$ will have full column rank. In the following we show that any $q$ rows of $\bF^{\perp}$ are linearly independent. To the contrary, suppose not; and let $\T'\subset[m]$ with $|\T'|=q$ be such that the $q\times q$ matrix $\bF_{\T'}^{\perp}$ is not a full rank matrix.
This implies that there exists a non-zero $\bc'\in\R^{q}$ such that $\bF_{\T'}^{\perp}\bc'=\bzero$. 
Let $\bb=\bF^{\perp}\bc'$. Note that $\bb\neq\bzero$ (because columns of $\bF^{\perp}$ are linearly independent)
and also that $\|\bb\|_0\leq m-q = k$. 
Now, since $\bF\bF^{\perp}=\bzero$, we have $\bF\bb=\bzero$, which contradicts the fact that any $k$ columns of $\bF$ are linearly independent.
\end{proof}
Now we design $\btS_i$'s. 
For $i\in[p]$, we set $\btS_i$ as follows:
\[
\btS_i=
\begin{bmatrix}
0 & \hdots & 0 & b_{11} & b_{21}  & \hdots & b_{l1} & 0 & \hdots & 0 \\
0 & \hdots & 0 & b_{12} & b_{22}  & \hdots & b_{l2} & 0 & \hdots & 0 \\
\vdots & \vdots & \vdots & \vdots & \vdots & \vdots & \vdots & \vdots & \vdots & \vdots \\
0 & \hdots & 0 & b_{1m} & b_{2m}  & \hdots & b_{lm} & 0 & \hdots & 0
\end{bmatrix}
\]
where $l=q$ if $i<p$; otherwise $l=\nr-(p-1)q$.
The first $(i-1)q$ and the last $\nr-[(i-1)q+l]$ columns of $\btS_i$ are zero. 
This also implies that the number of rows in each $\bS_i$ is $p=\lceil\nr/q\rceil$.

\begin{claim}\label{claim:F-ortho-S}
For every $i\in[p]$, we have $\bF\btS_i=0$.
\end{claim}
\begin{proof}
By construction, the null-space of $\bF$ is $\N(\bF)=\textsf{span}\{\bb_1,\bb_2,\hdots,\bb_q\}$, 
which implies that $\bF\bb_i=\bzero$, for every $i\in[q]$.
Since all the columns of $\btS_i$'s are either $\bzero$ or $\bb_j$ for some $j\in[q]$, the claim follows.
\end{proof}
The above constructed matrices $\btS_i$'s give the following encoding matrix $\bS_i$ for the $i$'th worker node:
\begin{align}\label{eq:encoding-matrix-S_i}
\bS_i=
\begin{bmatrix}
b_{1i} \hdots b_{qi} & & & \\ 
& \ddots & & \\
& & b_{1i} \hdots b_{qi} & \\
& & & b_{1i} \hdots b_{li}
\end{bmatrix}_{p\times \nr}
\end{align}
All the unspecified entries of $\bS_i$ are zero. The matrix $\bS_i$ is for encoding the data for worker $i$.
By stacking up the $\bS_i$'s on top of each other gives us our desired encoding matrix $\bS$.

To get efficient encoding, we want $\bS$ to be as sparse as possible. \label{encoding-rref}
Since $\bS$ is completely determined by $\bF^{\perp}$, whose columns are the basis vectors of $\N(\bF)$, 
it suffices to find a sparse basis for $\N(\bF)$. It is known that finding the sparsest basis for the null-space of a matrix is NP-hard \cite{ColemanPo86}. 
Note that we can always find the basis vectors of $\N(\bF)$ by reducing $\bF$ to its row-reduced-echelon-form (RREF) 
using the Gaussian elimination \cite{HoffmanKu71_LAbook}. 
This will result in $\bF^{\perp}$ whose last $q$ rows forms a $q\times q$ identity matrix.
Note that  $q=m-k$, where $k=2t$. So, if the corruption threshold $t$ is very small as compared to $m$, 
the $\bF^{\perp}$ that we obtain by the RREF will be very sparse -- only the first $2t$ rows may be dense.
Since computing $\bS$ is equivalent to computing $\bF^{\perp}$, and we can compute $\bF^{\perp}$ in $O(k^2m)$ time
using the Gaussian elimination, the time complexity of computing $\bS$ is also $O(k^2m)$.

Now we prove an important property of the encoding matrix $\bS$ that will be crucial for recovery of the desired matrix-vector product.
\begin{claim}\label{claim:S_T-full-rank}
For any $\T\subset[m]$ such that $|\T|\geq (m-t)$, let $\bS_{\T}$ denote the $|\T|p\times \nr$ matrix obtained from $\bS$ by restricting it to all the blocks $\bS_i$'s for which $i\in\T$.
Then $\bS_{\T}$ is of full column rank.
\end{claim}
\begin{proof}
For $i\in[p-1]$, let $\B_i=[(i-1)q+1 : iq]$ and $\B_p=[(p-1)q+1: \nr-(p-1)q]$, 
where we see $\B_i$'s as a collection of some column indices.
Consider any two distinct $i,j\in[p]$. It is clear that for any two vectors $\bu_1\in\B_i,\bu_2\in\B_j$, we have 
$\supp(\bu_1)\cap\supp(\bu_2)=\phi$, which means that all the columns in distinct $\B_i$'s are linearly independent.
So, to prove the claim, we only need to show that the columns within the same $\B_i$'s are linearly independent.
Fix any $i\in[p]$, and consider the $|\T|p\times q$ sub-matrix $\bS_{\T}^{(i)}$ of $\bS_{\T}$,
which is obtained by restricting $\bS_{\T}$ to the columns in $\B_i$. 
There are precisely $|\T|$ non-zero rows in $\bS_{\T}^{(i)}$, which are equal to the rows of the matrix $\bF_{\T}^{\perp}$ defined in \Claimref{rest-F-perp_full-rank}. We have already shown in the proof of \Claimref{rest-F-perp_full-rank} that $\bF_{\T}^{\perp}$ is of full column rank. Therefore, $\bS_{\T}^{(i)}$ is also of full column rank. This concludes the proof of \Claimref{S_T-full-rank}.
\end{proof}
Since $\bS_{\T}$ is of full column rank, in principle, we can recover any vector $\bu\in\R^{\nr}$ from $\bS_{\T}\bu$. 
In the next section, we show an efficient way for this recovery.

\subsection{Recovering The Matrix-Vector Product $\bA\bv$}\label{subsec:recovery_matvecprod}
Once the master has found the set $\I$ of corrupt worker nodes, it discards all the data received from them.
Let $\T=[m]\setminus\I=\{i_1,i_2,\hdots,i_f\}$ be the set of all honest worker nodes, where $f=(m-|\I|)\geq(m-t)$.
Let $\br=[\br_1^T \br_2^T \hdots \br_m^T]$, 
where $\br_i=\bS_i\bA\bv+\be_i$. 
All the $\br_i$'s from the honest worker nodes can be written as
\begin{equation}\label{eq:error-free-system}
\br_{\T}=\bS_{\T}\bA\bv, 
\end{equation}
where $\bS_{\T}$ is as defined in \Claimref{S_T-full-rank}, and $\br_{\T}$ is also defined analogously and equal to the restriction of $\br$ to all the $\br_i$'s for which $i\in\T$. 
Since $\bS_{\T}$ has full column rank (by \Claimref{S_T-full-rank}), in principle, we can recover $\bA\bv$ from \eqref{eq:error-free-system}. 
Next we show how to recover $\bA\bv$ efficiently, by exploiting the structure of $\bS$.

Let $\btr_j=[r_{i_1j},r_{i_2j},\hdots,r_{i_fj}]^T$, for every $j\in[p]$. 
The repetitive structure of $\bS_i$'s (see \eqref{eq:encoding-matrix-S_i}) allows us to write \eqref{eq:error-free-system} equivalently in terms of $p$ smaller systems.
\begin{align}
\btr_j &= \bF_j(\bA\bv)_{\B_j},\quad \text{for }j\in[p], \label{eq:matvecprod-1}
\end{align}
where, for $j\in[p-1]$, $\B_i=[(i-1)q+1 : iq]$ and $\bF_j = \bF_{\T}^{\perp}$, and $\B_p=[(p-1)q+1: \nr-(p-1)q]$ and $\bF_p$ is equal to the restriction of $\bF_{\T}^{\perp}$ to its first $(\nr-(p-1)q)$ columns.
Since $\bF_{\T}^{\perp}$ has full column rank (by \Claimref{rest-F-perp_full-rank}), we can compute $(\bA\bv)_{\B_i}$ for all $i\in[p]$, by multiplying \eqref{eq:matvecprod-1} by $\bF_j^+=(\bF_j^T\bF_j)^{-1}\bF_j^T$, which it called the Moore-Penrose inverse of $\bF_j$.
Since $\bA\bv=[(\bA\bv)_{\B_1}^T,(\bA\bv)_{\B_2}^T,\hdots,(\bA\bv)_{\B_p}^T)]^T$,
we can recover the desired MV product $\bA\bv$.

\paragraph{Time Complexity analysis.}
The task of obtaining $\bA\bv$ from $\bS_{\T}\bA\bv$ reduces to 
{\sf (i)} computing $\bF_j^+=(\bF_{\T}^{\perp})^+$ once, which takes $O(q^2|\T|)$ time na\"ively; 
{\sf (ii)} computing $\bF_p^+$ once, which takes at most $O(q^2|\T|)$ time na\"ively; and 
{\sf (iii)} computing the MV products $\bF_j^+\btr_j$ for every $j\in[p]$, which takes $O(pq|\T|)$ time in total.
Since $p$ is much bigger than $q$, the total time taken in recovering $\bA\bv$ from $\bS_{\T}\bA\bv$ is $O(pq|\T|)=O(pm^2)$.

\subsection{Designing The Error Locator Matrix $\bF$}\label{subsec:F-matrix}
In this section, we design a $k\times m$ matrix $\bF$ (where $k<m$) such that for any {\em sparse} error vector $\be\in\R^m$, we can uniquely and efficiently recover $\be$ (and, therefore, $\supp(\be)$) from the under-determined system of linear equations $\bbf=\bF\be\in\R^k$. 
This is related to the {\em sparse representation problem}, where one would like to 
find the sparsest representation of $\bbf$ in terms of the linear combination of the columns of $\bF$, i.e., minimizing $\|\be\|_0$ subject to the constraint that $\bbf=\bF\be$. This problem is of combinatorial nature and is known to be NP-hard \cite{CandesTao05}.
To make this problem computationally tractable, Candes and Tao \cite{CandesTao05} showed that if $\bF$ satisfies a certain regularity condition (which they named the {\em restricted isometry property} (RIP)), then the sparsest reconstruction problem can be reduced to minimizing $\|\be\|_1:=\sum_{i=1}^m |e_i|$ subject to the constraint that $\bbf=\bF\be$, which can be efficiently solved using a linear program. They also showed that a random Gaussian matrix satisfies the RIP condition.
A common problem with such random constructions is that they may not work with small block-lengths (in our setting, $m$ is the number of workers which may not be a big number), and can only correct a constant fraction of errors, where the constant is very small.
We need a deterministic construction that can handle a constant fraction (ideally up to 1/2) of errors and that works with small block-lengths.

Ak{\c{c}}akaya and Tarokh \cite{AkcakayaTa08} proposed an efficient solution to the sparse representation problem using {\em Vandermonde} matrices. To construct them, take $m$ distinct non-zero elements $z_1,z_2,\hdots,z_m$ from $\R$, and consider the following $k\times m$ Vandermonde matrix $\bF$.
\begin{align}\label{eq:vandermonde}
\bF = 
\begin{bmatrix}
1 & 1 & 1 & \hdots & 1 \\
z_1 & z_2 & z_3 & \hdots & z_m \\
z_1^2 & z_2^2 & z_3^2 & \hdots & z_m^2 \\
\vdots & \vdots & \vdots & \ddots & \vdots \\
z_1^{k-1} & z_2^{k-1} & z_3^{k-1} & \hdots & z_m^{k-1} \\
\end{bmatrix}_{k\times m}
\end{align}
For the above $\bF$, it was shown in \cite{AkcakayaTa08} that, if $|\supp(\be)|\leq k/2$, then the Reed-Solomon type decoding can be used for exact reconstruction of $\be$ 
from $\bbf=\bF\be$.\footnote{Note that, since any $k$ columns of $\bF$ (which is the Vandermonde matrix) are linearly independent, if there exists a vector $\be$ such that $|\supp(\be)|\leq k/2$ and $\be$ satisfies $\bbf=\bF\be$ for a fixed $\bbf$, then $\be$ is unique.} Furthermore, their decoding algorithm is efficient and runs in $O(m^2)$ time. 
The results in \cite{AkcakayaTa08} are given for complex vector spaces, and they hold over real numbers also. Below we state the sparse recovery result (specialized to reals) from \cite{AkcakayaTa08}.

\begin{lem}[\hspace{-0.001cm}\cite{AkcakayaTa08}]\label{lem:RS-sparse-decoding}
Let $\bF$ be the $k\times m$ matrix as defined in \eqref{eq:vandermonde}. Let $\be\in\R^m$ be an arbitrary vector with $|\supp(\be)|\leq k/2$. We can exactly recover the vector $\be$ from $\bbf = \bF\be$ in $O(m^2)$ time.
\end{lem}
Note that $\bF$ is a $k\times m$ matrix, where $k<m$. Choosing $k$ is in our hands, and larger the $k$, more the number of errors we can correct (but at the expense of increased storage and computation); see \Subsectionref{resource_GD} for more details.

\subsection{Resource Requirement Analysis}\label{subsec:resource_GD}
In this section, we analyze the total amount of resources (storage, computation, and communication) required by our method for computing gradients in the presence of $t$ (out of $m$) adversarial worker nodes 
and prove \Theoremref{main-result}.
Fix an $\epsilon>0$. Let the corruption threshold $t$ satisfy $t\leq \lfloor(\epsilon/(1+\epsilon))\cdot(m/2)\rfloor$. 

As described earlier in \Subsectionref{approach_GD}, 
we compute the gradient $\nabla f(\bw)=\bX^Tf'(\bw)$ in two-rounds; and in each round we use the Byzantine-tolerant MV multiplication, which we have developed in \Sectionref{matrix-vector-mult}, as a subroutine; 
see \Figureref{problem-setup} for a pictorial representation of our scheme.
We encode $\bX$ to compute $f'(\bw)$ in the 1st round: first compute $\bX\bw$ using MV multiplication and then locally compute $f'(\bw)$. 
To compute $\bX^Tf'(\bw)$ (which is equal to the gradient) in the 2nd round, we encode $\bX^T$ and compute
$\bX^Tf'(\bw)$. Let $\bS^{(1)}$ and $\bS^{(2)}$ be the encoding matrices of dimensions $p_1m\times n$ and $p_2m\times d$, respectively, to encode $\bX$ and $\bX^T$, respectively.
Here, $p_1=\lceil n/q\rceil$ and $p_2=\lceil d/q\rceil$, where $q=m-k$. 
Since $k=2t$ (by \Lemmaref{RS-sparse-decoding}), we have $q=(m-k)\geq m/(1+\epsilon)$. 

\subsubsection{Storage Requirement}
Each worker node $i$ stores two matrices $\bS_i^{(1)}\bX$ and $\bS_i^{(2)}\bX^T$.
The first one is a $p_1\times(d+1)$ matrix, and the second one is a $p_2\times n$ matrix.
So, the total amount of storage at all worker nodes is equal to storing $(p_1(d+1)+p_2n)\times m$ real numbers. 
Since $p_1 \leq \lceil(1+\epsilon)\frac{n}{m}\rceil$ and $p_2 \leq \lceil(1+\epsilon)\frac{d}{m}\rceil$, 
the total storage is 
\begin{align*}
\big(p_1(d+1)&+p_2n\big)m = p_1m(d+1) + p_2mn \\
&< [(1+\epsilon)n + m](d+1) + [(1+\epsilon)d +m]n \\
&= (1+\epsilon)n(2d+1) + m(n+d+1).
\end{align*}
where the first term is roughly equal to a $2(1+\epsilon)$ factor more than the size of $\bX$. Note that the second term does not contribute much to the total storage as compared to the first term, because the number of worker nodes $m$ is much smaller than both $n$ and $d$. In fact, if $m-k$ divides both $n$ and $d$, then the second term vanishes.
Since $|\bX|$ is an $n\times d$ matrix, the total storage at each worker node is almost equal to $2(1+\epsilon)\frac{|\bX|}{m}$, which is a constant factor of the optimal, that is, $\frac{|\bX|}{m}$, and the total storage is roughly equal to $2(1+\epsilon)|\bX|$.

\subsubsection{Computational Complexity}
We can divide the computational complexity of our scheme as follows:
\begin{itemize}[leftmargin = *]
\item {\it Encoding the data matrix.} 
Since, for every $i\leq k$ and $j>k$, the total number of non-zero entries in $\bS_i^{(1)}$ and $\bS_j^{(1)}$ are at most $n$ and $p_1$, respectively (see \Subsectionref{encoding-matrix} for details), 
the computational complexity for computing $\bS_i^{(1)}\bX$ for each $i\leq k$, and $\bS_j^{(1)}\bX$ for each $j>k$, is $O(nd)$ and $O(p_1d)$, respectively.
So, the encoding time for computing $\bS^{(1)}\bX$ is equal to $O\left(k(nd)+(m-k)(p_1d)\right)=O\left((\frac{\epsilon}{1+\epsilon}m+1)nd\right)$.
Similarly, we can show that the encoding time for computing $\bS^{(2)}\bX^T$ is also equal to $O\left((\frac{\epsilon}{1+\epsilon}m+1)nd\right)$.
Note that computing $\bS^{(1)}$ and $\bS^{(2)}$ take $O(k^2m)$ time each, which is much smaller, as compared to the encoding time. 
So, the total encoding time is $O\left((\frac{\epsilon}{1+\epsilon}m+1)nd\right)$.
Note that this encoding is to be done only once. 

\item {\it Computation at each worker node.} 
In the first round, upon receiving $\bw$ from the master node, 
each worker $i$ computes $(\bS_i^{(1)}\bX)\bw$, and reports back the resulting vector. 
Similarly, in the second round, upon receiving $f'(\bw)$ from the master node, 
each worker $i$ computes $(\bS_i^{(2)}\bX^T)f'(\bw)$, and reports back the resulting vector. 
Since $\bS_i^{(1)}\bX$ and $\bS_i^{(2)}\bX^T$ are $p_1\times(d+1)$ and $p_2\times n$ matrices, respectively, 
each worker node $i$ requires $O(p_1d+p_2n)=O((1+\epsilon)\frac{nd}{m})$ time.

\item {\it Computation at the master node.}
The total time taken by the master node in both the rounds is the sum of the time required in 
(i) finding the corrupt worker nodes in the 1st and 2nd rounds, which requires $O(p_1km)$ and $O(p_2km)$ time, respectively (see \Subsectionref{finding-corrupt-nodes}), 
(ii) recovering $\bX\bw$ from $\bS_{\T}^{(1)}\bX\bw$ in the 1st round, which requires $O(p_1m^2)$ time, 
(iii) computing $f'(\bw)$ from $\bX\bw$, which takes $O(n)$ time, and
(iv) recovering  $\bX^Tf'(\bw)$ from $\bS_{\T}^{(2)}\bX^Tf'(\bw)$ in the 2nd round, which requires $O(p_2m^2)$ time (see \Subsectionref{recovery_matvecprod}).
Since $k<m$, the total time is equal to $O((p_1+p_2)m^2)=O((1+\epsilon)(n+d)m)$. 
\end{itemize}

\subsubsection{Communication Complexity}
In each gradient computation, 
{\sf (i)} master broadcasts $(n+d)$ real numbers, $d$ in the first round and $n$ in the second round; and
{\sf (ii)} each worker sends $\left((1+\epsilon)\frac{n+d}{m}\right)$ real numbers to master, $(1+\epsilon)\frac{n}{m}$ in the first round and $(1+\epsilon)\frac{d}{m}$ in the second round.

%% file: solution_CD.tex
\section{Our Solution to Coordinate Descent}\label{sec:solution_CD}
In this section, we give a solution to the distributed coordinate descent (CD) under Byzantine attacks and prove \Theoremref{main-result_CD}.
To make our notation simpler, we remove the dependence on the label vector $\by$ in the problem expression \eqref{eq:problem-cd} and rewrite it as follows (this is without loss of generality in the light of \Footnoteref{storing-labels} and Algorithm 1):
\begin{align}\label{eq:problem-express}
\arg\min_{\bw\in\R^d}\phi(\bX\bw):=\sum_{i=1}^n \ell(\langle \bx_i,\bw\rangle).
\end{align}
We want to optimize \eqref{eq:problem-express} using distributed CD, described in \Subsectionref{setting-cd}.
As outlined in \Subsectionref{approach_CD}, we use data encoding and error correction over real numbers for that.
To combat the effect of adversary, we add redundancy to enlarge the parameter space.
Let $\btX^R=\bX\bR$, where $\bR=[\bR_1\ \bR_2\ \hdots\ \bR_m]\in\R^{d\times pm}$ with $pm\geq d$,
and each $\bR_i$ is a $p\times d$ matrix. 
We will determine the encoding matrix $\bR$ later, 
after describing what properties we want from it. 
For the value of $p$, looking ahead, when $t$ is the number of corrupt workers, we will choose $p=\frac{d}{m-2t}$, which is a constant multiple of $\frac{d}{m}$ even if $t$ is a constant fraction ($<\frac{1}{2}$) of $m$ (e.g., for $t=\frac{m}{3}$, we have $p=\frac{3d}{m}$).
We consider $\bR$'s which are of full row-rank. 
Let $\bR^+:=\bR^T(\bR\bR^T)^{-1}$ denote its Moore-Penrose inverse 
such that $\bR\bR^+=I_d$, where $I_d$ is the $d\times d$ identity matrix.
Note that $\bR^+$ is of full column-rank. 
Let $\bv=\bR^+\bw$ be the transformed vector, which lies in a larger (than $d$) dimensional space. 
Let $\bR^+=[(\bR_1^+)^T\ (\bR_2^+)^T\ \hdots\ (\bR_m^+)^T]^T$, where each $\bR_i^+:=(\bR^+)_i$ is a $p\times d$ matrix. 
With this, by letting $\bv=[\bv_1^T\ \bv_2^T\ \hdots\ \bv_m^T]^T$, we have that $\bv_i=\bR_i^+\bw$ for every $i\in[m]$.
Now, consider the following modified problem over the encoded data.
\begin{align}\label{eq:encoded-problem-express}
\arg\min_{\bv\in\R^{pm}}\phi(\btX^R\bv).
\end{align}
Observe that, since $\bR$ is of full row-rank, $\min_{\bw\in\R^d}\phi(\bX\bw)$ is equal to $\min_{\bv\in\R^{pm}}\phi(\btX^R\bv)$; and from an optimal solution to one problem we can obtain an optimal solution to the other problem. 
We design an encoding/decoding scheme such that when we optimize the encoded problem \eqref{eq:encoded-problem-express} using \Algorithmref{distr-coor-desc-algo}, the vector $\bv$ that we get in each iteration is of the form $\bv=\bR^+\bw$ for some vector $\bw\in\R^d$.\footnote{If such a $\bw$ exists, then it is unique. This follows from the fact that $\bR^+$ is of full column-rank.} 
In fact, our encoding/decoding will ensure that the $\bw$ for which $\bv=\bR^+\bw$ would be equal to the original parameter vector in that iteration if we had run \Algorithmref{distr-coor-desc-algo} to solve \eqref{eq:problem-express}. We need this property because in any CD iteration $t$, we need access to the original parameter vector $\bw^t$ (such that $\bv^t=\bR^+\bw^t$) to facilitate the local parameter updates of $\bv_1^t,\hdots,\bv_m^t$ at the workers. See the paragraph after \eqref{eq:partial-update} for more details.

Now, instead of solving \eqref{eq:problem-express}, we solve its encoded form \eqref{eq:encoded-problem-express} using \Algorithmref{distr-coor-desc-algo} (with decoding at the master), 
where each worker $i$ stores $\btX_i^R=\bX\bR_i$ and is responsible for updating (some coordinates of) $\bv_i$. 
In the following, let $\U\subseteq[p]$ be a fixed arbitrary subset of $[p]$. 
Let $\bv^0:=\bR^+\bw^0$ for some $\bw^0$ at time $t=0$.
Suppose, at the beginning of the $t$'th iteration, we have $\bv^t=\bR^+\bw^t$ for some $\bw^t$, and each worker $i$ updates $\bv_{i\U}^t$ according to 
\begin{align}\label{eq:encoded-update-rule}
\bv_{i\U}^{t+1} = \bv_{i\U}^t - \alpha_t\nabla_{i\U} \phi(\btX^R\bv^t),
\end{align}
where $\nabla_{i\U} \phi(\btX^R\bv^t) = (\btX_{i\U}^R)^T\phi'(\btX^R\bv^t)$.
Recall that each $\bR_i$ is a $d\times p$ matrix, and each $\bR_i^+:=(\bR^+)_i$ is a $p\times d$ matrix.
We denote by $\bR_{i\U}$ the $d\times|\U|$ matrix obtained by restricting the columns of $\bR_i$ to the elements of $\U$.
Analogously, we denote by $\bR_{i\U}^+:=(\bR^+)_{i\U}$ the $|\U|\times d$ matrix obtained by restricting the rows of $\bR_i^+$ to the elements of $\U$.
With this, we can write $\btX_{i\U}^R=\bX\bR_{i\U}$. Now,
\eqref{eq:encoded-update-rule} can be equivalently written as
\begin{align}
\bv_{i\U}^{t+1} =  \bv_{i\U}^t - \alpha_t \bR_{i\U}^T\bX^T\phi'(\btX^R\bv^t). \label{eq:partial-update}
\end{align}
In order to update $\bv_{i\U}^t$, worker $i$ requires $\phi'(\btX^R\bv^t)$, where $\btX^R\bv^t=\sum_{j=1}^m\btX_j^R\bv_j^t$ and worker $i$ has only $(\btX_i^R,\bv_i^t)$.
Since $\bv^t=\bR^+\bw^t$, we have $\btX^R\bv^t=\bX\bR\bv^t=\bX\bw^t$. 
So, it suffices to compute $\bX\bw^t$ at the master node -- once master has $\bX\bw^t$, it can locally compute $\phi'(\bX\bw^t)$ and send it to all the workers.
Computing $\bX\bw^t$ is the distributed matrix-vector (MV) multiplication problem, where the matrix $\bX$ is fixed and we want to compute $\bX\bw^t$ for any vector $\bw^t$ in the presence of an adversary.
In \Sectionref{matrix-vector-mult}, we give a method for performing distributed MV multiplication in the presence of an adversary.
Now we give an overview, together-with an improvement on its computational complexity.

We encode $\bX$ using an encoding matrix $\bL\in\R^{(p'm)\times n}$.
Let $\bL=[\bL_1^T\ \bL_2^T\ \hdots\ \bL_m^T]^T$, where each $\bL_i$ is a $p'\times n$ matrix with $p'=\lceil\frac{n}{m-2t}\rceil$. 
Each $\bL_i$ has $p'$ rows and $n$ columns, and has the same structure as that of $\bS_i$ from \eqref{eq:encoding-matrix-S_i}. 
Worker $i$ stores $\btX_i^L=\bL_i\bX$.
To compute $\bX\bw$, master sends $\bw$ to all the workers; worker $i$ responds with $\bL_i\bX\bw+\be_i$, where $\be_i=\bzero$ if the $i$'th worker is honest, otherwise can be arbitrary; upon receiving $\{\bL_i\bX\bw+\be_i\}_{i=1}^m$, where at most $t$ of the $\be_i$'s can be non-zero, master applies the decoding procedure and recovers $\bX\bw$ back.
We can improve the computational complexity of this method significantly by observing that, in each iteration of our distributed CD algorithm, 
only a few coordinates of $\bw$ get updated and the rest of the coordinates remain unchanged. 
(Looking ahead, when each worker updates $\bv_{i\U}$'s according to \eqref{eq:encoded-update-rule}, it automatically updates $\bw_{f(\U)}$ according to \eqref{eq:dist-coor-dist-update} -- for a specific function $f$ as defined in \eqref{eq:function_CD} -- where $\bv$ and $\bw$ satisfy $\bv=\bR^+\bw$.)
This implies that for computing $\bX\bw$, master only needs to send the updated coordinates to the workers and keeps the result from the previous MV product with itself.
This significantly reduces the local computation at the worker nodes, as now they only need to perform a local MV product of a matrix of size $p'\times |f(\U)|$ and a vector of length $|f(\U)|$. See \Sectionref{matrix-vector-mult} for details.

Our goal in each iteration of CD is to update some coordinates of the original parameter vector $\bw$; instead, 
by solving the encoded problem, we are updating coordinates of the transformed vector $\bv$. 
We would like to design an algorithm/encoding such that it has exactly the same convergence properties as if we are running the distributed CD on the original problem without any adversary.
For this, naturally, we would like our algorithm to satisfy the following property:

{\em Update on any (small) subset of coordinates of $\bw$ should be achieved by updating some (small) subset of coordinates of $\bv_i$'s; and, by updating those coordinates of $\bv_i$'s, we should be able to efficiently recover the correspondingly updated coordinates of $\bw$. Furthermore, this should be doable despite the errors injected by the adversary in every iteration of the algorithm.}

Note that if each coordinate of $\bv$ depends on too many coordinates of $\bw$, then updating a few coordinates of $\bv$ may affect many coordinates of $\bw$, and it becomes information-theoretically impossible to satisfy the above property 
(even without the presence of an adversary).\footnote{To see this, consider the case when each worker $i$ updates only the first coordinate of $\bv_i$ and no worker is corrupt. Master receives $m$ linear equations $\bv_{i1}=\bR^+_{i1}\bw,\ i=1,2,\hdots,m$, where $\bR^+_{i1}$ is the first row of $\bR^+_i$ for every $i\in[m]$. Assume, for simplicity, that these $m$ equations are linearly independent. When $m$ is smaller than $d$ (which is always the case), there are infinite solutions to this system of linear equations, unless at most $m$ elements of $\bw$ are involved in the $m$ linear equations (i.e., the number of unknowns are at most the number of equations), which is equivalent to saying that the rows $\bR^+_{i1}$ for $i=1,2,\hdots,m$ are sparse. Our encoding matrix will satisfy this property; see \Subsectionref{encoding-decoding} for more detail.}
This imposes a restriction that each row of $\bR^+$ must have few non-zero entries, in such a way that updating $\bv_{i\U}^t$'s, for any choice of $\U\subseteq[p]$, will collectively update only a subset (which may potentially depend on $\U$) of coordinates of the original parameter vector $\bw^t$, and we can uniquely and efficiently recover those updated coordinates of $\bw^t$, even from the {\em erroneous} vectors $\{\bv_{i\U}^{t+1}+\be_{i\U}\}_{i=1}^m$, where at most $t$ out of $m$ error vectors $\{\be_{i\U}\}_{i=1}^m$ are non-zero and may have arbitrary entries.
In order to achieve this, we will design a sparse encoding matrix $\bR^+$ (which in turn determines $\bR$), that satisfies the following properties:
\begin{figure*}[t]
\centering
\begin{tikzpicture}[>=stealth', font=\sffamily\large\bfseries, thick]
\draw [->] (-0.8,4.0) -- (1.5,4.0); \node at (0.4,4.3) {$\bar{\bar{\bw}}_{f(\U')}^t$};
\node [rotate=0, scale=0.75] at (2.5,4.8) {M \text{broadcasts} $\bar{\bar{\bw}}_{f(\U')}^t$};
\draw (1.5,3.5) rectangle (3.5,4.5) node [pos=0.5] {M};
\draw [->] (3.5,4.0) -- (4.5,4.0);
\draw (4.5,3.5) rectangle (5.5,4.5) node [pos=0.5] {Dec};

\draw (0,0.5) rectangle (0.9,1.4) node [pos=0.5] {$W_1$}; \node at (0.45, 0.1) {$\bL_1\bX$};
\draw [fill=red] (1.8,0.5) rectangle (2.7,1.4) node [pos=0.5] {$W_2$}; \node at (2.25, 0.1) {$\bL_2\bX$};
\draw (3.6,0.5) rectangle (4.5,1.4) node [pos=0.5] {$W_3$}; \node at (4.05, 0.1) {$\bL_3\bX$};
\draw [fill=black] (5.2,0.95) circle [radius=0.04];
\draw [fill=black] (5.5,0.95) circle [radius=0.04];
\draw [fill=black] (5.8,0.95) circle [radius=0.04];
\draw [fill=red] (6.5,0.5) rectangle (7.4,1.4) node [pos=0.5] {$W_m$}; \node at (6.95, 0.1) {$\bL_m\bX$};

\draw [->] (0.5,1.4) -- (1.6,3.5); \node [rotate=62, scale=0.7] at (0.85,2.5) {$\bL_1\bX\bar{\bar{\bw}}_{f(\U')}^t$};
\draw [->] (2.25,1.4) -- (2.25,3.5); \node [rotate=90, scale=0.7] at (2.05,2.45) {$\bL_2\bX\bar{\bar{\bw}}_{f(\U')}^t+\be_2$};
\draw [->] (4.05,1.4) -- (2.75,3.5); \node [rotate=-60, scale=0.7] at (3.7,2.35) {$\bL_3\bX\bar{\bar{\bw}}_{f(\U')}^t$};
\draw [->] (6.95,1.4) -- (3.4,3.5); \node [rotate=-30, scale=0.7] at (5.5,2.5) {$\bL_m\bX\bar{\bar{\bw}}_{f(\U')}^t+\be_m$};

\draw [->] (5.5,4.0) -- (7.0,4.0); 
\node [scale=0.9] at (6.25,4.3) {$\bX\bw^t$};
\draw (7.0,3.5) rectangle (9.0,4.5);
\node [scale=0.9] at (8.0,4.2) {Compute};
\node [scale=0.8] at (8.0,3.8) {$\phi'(\bX\bw^t)$};
\draw [->] (9.0,4.0) -- (10.7,4.0); 
\node [scale=0.8] at (9.85,4.3) {$\phi'(\bX\bw^t)$};

\node [rotate=0, scale=0.75] at (11.7,4.8) {M \text{broadcasts} $\phi'(\bX\bw^t)$};
\draw (10.7,3.5) rectangle (12.7,4.5) node [pos=0.5] {M};
\draw [->] (12.7,4.0) -- (13.7,4.0);
\draw (13.7,3.5) rectangle (14.7,4.5) node [pos=0.5] {Dec};

\draw (9.2,0.5) rectangle (10.1,1.4) node [pos=0.5] {$W_1$}; \node at (9.65, 0.1) {$\bX\bR_1$};
\draw (11.0,0.5) rectangle (11.9,1.4) node [pos=0.5] {$W_2$}; \node at (11.45, 0.1) {$\bX\bR_2$};
\draw [fill=red] (12.8,0.5) rectangle (13.7,1.4) node [pos=0.5] {$W_3$}; \node at (13.25, 0.1) {$\bX\bR_3$};
\draw [fill=black] (14.4,0.95) circle [radius=0.04];
\draw [fill=black] (14.7,0.95) circle [radius=0.04];
\draw [fill=black] (15.0,0.95) circle [radius=0.04];
\draw [fill=red] (15.7,0.5) rectangle (16.6,1.4) node [pos=0.5] {$W_m$}; \node at (16.15, 0.1) {$\bX\bR_m$};

\draw [->] (9.7,1.4) -- (10.8,3.5); \node [rotate=62, scale=0.9] at (9.95,2.5) {$\bv_{1\U}^{t+1}$};
\draw [->] (11.45,1.4) -- (11.45,3.5); \node [rotate=90, scale=0.9] at (11.15,2.25) {$\bv_{2\U}^{t+1}$};
\draw [->] (13.25,1.4) -- (11.95,3.5); \node [rotate=-58, scale=0.9] at (13.0,2.35) {$\bv_{3\U}^{t+1}+\be_{3\U}$};
\draw [->] (16.15,1.4) -- (12.6,3.5); \node [rotate=-30, scale=0.9] at (14.75,2.6) {$\bv_{m\U}^{t+1}+\be_{m\U}$};

\draw [->] (14.7,4.0) -- (17.0,4.0); \node [scale=0.9] at (15.85,4.3) {$\bw_{f(\U)}^{t+1}$};

\draw [->] (17.0,4.0) -- (17.0,6.5); \draw [->] (17.0,6.5) -- (-0.8,6.5); \draw [->] (-0.8,6.5) -- (-0.8,4.0);

\node at (8.0,6.9) {$t\leftarrow t+1;\ \U'\leftarrow\U;\ \bar{\bar{\bw}}_{f(\U')}^t :=\bw_{f(\U')}^{t-1}-\bw_{f(\U')}^t$};

\end{tikzpicture}
\caption{This figure shows our 2-round approach to the Byzantine-resilient distributed coordinate descent (CD) for solving \eqref{eq:problem-express} using data encoding and real-error correction. 
We encode $\bX$ with the encoding matrix $[\bR_1\ \hdots\ \bR_m]\in\R^{d\times p_2m}$ and store $\btX_i^R:=\bX\bR_i$ at the $i$'th worker and solve \eqref{eq:encoded-problem-express} over an enlarged parameter vector $\bv\in\R^{p_2m}$.
At the $t$'th iteration, for some $\U\subseteq[p_2]$, the update at the $i$'th worker is $\bv_{i\U}^{t+1} =  \bv_{i\U}^t - \alpha_t \bR_{i\U}^T\bX^T \phi'(\btX^R\bv^t)$, which requires $\phi'(\btX^R\bv^t)$, where $\btX^R\bv^t=\bX\bw^t$. The first part of the figure is for providing $\phi'(\bX\bw^t)$ to every worker in each iteration so that they can update $\bv_{i\U}^t$'s. 
For this, we encode $\bX$ using the encoding matrix $[\bL_1^T\ \hdots\ \bL_m^T]^T\in\R^{p_1m\times n}$ and store $\btX_i^L:=\bL_i\bX$ at worker $i$. The encoding has the property that we can recover $\bX\bw^t$ from the erroneous vectors $\{\btX_i^L\bw^t+\be_i\}_{i=1}^m$, where at most $t$ of the $\be_i$'s are non-zero and can be arbitrary.
We can make it computationally more efficient at the workers' side by observing that, in each iteration, only a subset of coordinates of $\bw$ are being updated: 
suppose we updated $\bv_{i\U'}^{t}$'s in the $t$'th iteration, which automatically updated $\bw_{f(\U')}^{t}$.
Since $\bw_{[d]\setminus f(\U')}^{t}$ remain unchanged, we need to send only $\bw_{f(\U')}^t$ to the workers -- in the figure, to take care of a technicality, we let master broadcast $\bar{\bar{\bw}}_{f(\U')}^t:=\bw_{f(\U')}^{t-1}-\bw_{f(\U')}^t$, each worker $i$ computes $\btX_i\bar{\bar{\bw}}_{f(\U')}^t$ and sends it backs to the master. Since master keeps $\bX\bw^{t-1}$ from the previous iteration with itself, it can compute $\bX\bw^t$.
The set of corrupt workers may be different in different rounds -- the corrupt ones are shown in red color and they can send arbitrary outcomes to master. Once master has recovered $\bX\bw^t$, it computes $\phi'(\bX\bw^t)$ and broadcasts it; upon receiving it worker $i$ updates $\bv_{i\U}^{t+1}$ and sends it back. By {\bf P.\ref{prop1}}, this reflects an update on $\bw_{f(\U)}^{t+1}$ according to \eqref{eq:update-w}; 
and by {\bf P.\ref{prop2}}, the master can recover $\bw_{f(\U)}^{t+1}$.}
\label{fig:problem-setup_CD}
\end{figure*}

\begin{enumerate}\renewcommand{\labelenumi}{\textbf{P.\theenumi}}
\item \label{prop1} $\bR^+$ has structured sparsity, which induces a map $f:[p]\to \mathcal{P}([d])$ (where $\mathcal{P}([d])$ denotes the power set of $[d]$) such that 
\begin{enumerate}
\item $\{f(i):i\in[p]\}$ partitions $\{1,2,\hdots,d\}$, i.e., for every $i,j\in[p]$, such that $i\neq j$, we have $f(i)\cap f(j)=\emptyset$ and that $\bigcup_{i=1}^pf(i)=[d]$.
\item $|f(i)|=|f(j)|$ for every $i,j\in[p-1]$, and $|f(p)|\leq |f(i)|$, for any $i\in[p-1]$. 
\item For any $\U\subseteq[p]$, define $f(\U):=\cup_{j\in\U}f(j)$. If we update $\bv_{i\U}^t$, $\forall i\in[m]$, according to \eqref{eq:partial-update}, 
it automatically updates $\bw_{f(\U)}^t$ according to 
\begin{equation}\label{eq:update-w}
\bw_{f(\U)}^{t+1}= \bw_{f(\U)}^t - \alpha_t\bX_{f(\U)}^T\phi'(\bX\bw^t).
\end{equation} 
If we set $\bv_{i\barU}^{t+1}:=\bv_{i\barU}^t$ and $\bw_{\barfU}^{t+1}:=\bw_{\barfU}^{t}$, then $\bv^{t+1}=\bR^+\bw^{t+1}$, i.e., our invariant holds.
\end{enumerate}
\end{enumerate}
Note that \eqref{eq:update-w} is the same update rule if we run the plain CD algorithm to update $\bw_{f(\U)}$. 
In fact, our encoding matrix satisfies a stronger property, that $\bv_{i\U}^{t+1}=\bR_{i\U,f(\U)}^+\bw_{f(\U)}^{t+1}$ holds for every $i\in[m]$, $\U\subseteq[p]$, where $\bR_{i\U,f(\U)}^+$ denotes the $|\U|\times|f(\U)|$ matrix obtained from $\bR_{i\U}^+$ by restricting its column indices to the elements in $f(\U)$. 
\begin{enumerate}\renewcommand{\labelenumi}{\textbf{P.\theenumi}}
\setcounter{enumi}{1}
\item \label{prop2} We can efficiently recover $\bw_{f(\U)}^{t+1}$ from the erroneous vectors $\{\bv_{i\U}^{t+1}+\be_{i\U}\}_{i=1}^m$, where at most $t$ of $\be_{i\U}$'s are non-zero and may have arbitrary entries.
Since $\bv_{i\U}^{t+1}=\bR_{i\U,f(\U)}^+\bw_{f(\U)}^{t+1}$, for every $i\in[m]$, $\U\subseteq[p]$, this property requires that not only $\bR^+$, but its sub-matrices also have error correcting capabilities.
\end{enumerate}
\begin{remark}
Note that {\bf P.\ref{prop1}} implies that for every $i\in[p]$, we have $|f(i)|\leq d/p$. As we see later, this will be equal to $m/(1+\epsilon)$ for some $\epsilon>0$ which is determined by the corruption threshold. 
This means that in each iteration of the CD algorithm running on the modified encoded problem, we will be effectively updating the coordinates of the parameter vector $\bw$ in chunks of size $m/(1+\epsilon)$ or its integer multiples. 
In particular, if each worker $i$ updates $k$ coordinates of $\bv_i$, then $km/(1+\epsilon)$ coordinates of $\bw$ will get updated.
For comparison, Algorithm 1 updates $km$ coordinates of the parameter vector $\bw$ in each iteration, if each worker updates $k$ coordinates in that iteration.
\end{remark}
Now we design an encoding matrix $\bR^+$ and a decoding method that satisfy {\bf P.\ref{prop1}} and {\bf P.\ref{prop2}}.

\subsection{Encoding and Decoding}\label{subsec:encoding-decoding}
In this section, we first design an encoding matrix $\bR^+$ that satisfies {\bf P.\ref{prop1}}.
$\bR^+$ will be such that it has orthonormal rows, so, $\bR$ is easy to compute, $\bR=(\bR^+)^T$.
For simplicity, we denote $\bR^+$ by $\bS$.
We show that the encoding matrix that we design for the MV multiplication in \Sectionref{matrix-vector-mult} satisfies all the properties that we want.\footnote{The encoding and decoding of this section is based on the corresponding algorithms from \Sectionref{matrix-vector-mult}.}
In the MV multiplication, we had a fixed matrix $\bA$ and the master node wants to compute $\bA\bw$ for any vector $\bw$ of its choice. In the solution presented in \Sectionref{matrix-vector-mult}, we encode $\bA$ and store $\bS_i\bA$ at the $i$'th worker node. Now, the master sends $\bw$ to all the worker nodes, and each worker $i$ responds with $\bS_i\bA\bw+\be_i$, where $\be_i=\bzero$ if worker $i$ is honest, otherwise can be arbitrary. Once master receives $\{\bS_i\bA\bw+\be_i\}_{i=1}^m$, it can run the error correcting procedure to recover $\bA\bw$.
To apply this in our setting, we take $\bA$ to be the identity matrix, such that $\bS_i\bA=\bS_i$, and the master can recover $\bw$ from $\{\br_i=\bS_i\bw+\be_i\}_{i=1}^m$, if at most $t$ of the $\be_i$'s are non-zero. 
For convenience, we rewrite the encoding matrix $\bS_i$ for the $i$'th worker node from \Subsectionref{encoding-matrix} below:
\begin{align}\label{eq:encoding-matrix-S_i_CD}
\bS_i=
\begin{bmatrix}
b_{1i} \hdots b_{qi} & & & \\ 
& \ddots & & \\
& & b_{1i} \hdots b_{qi} & \\
& & & b_{1i} \hdots b_{li}
\end{bmatrix}_{p\times d}
\end{align}
Here $q=(m-2t)$ and $l=d-(p-1)q$, where $p=\lceil\frac{d}{q}\rceil$. Note that $1\leq l<q$, and if $q$ divides $d$, then $l=q$.
All the unspecified entries of $\bS_i$ are zero. By stacking up the $\bS_i$'s gives us our desired encoding matrix $\bS=[\bS_1^T\ \bS_2^T\ \hdots\ \bS_m^T]^T$.
Note that $b_{1i},b_{2i},\hdots,b_{qi}$ are such that if we let $\bb_i=[b_{i1}\ b_{i2} \hdots b_{im}]^T$ for every $i\in[q]$, then $\{\bb_1, \bb_2,\hdots,\bb_q\}$ is a set of orthonormal vectors. This implies that $\bS$ is orthonormal, and, therefore, $\bS^+=\bS^T$. 
By taking $\bR=\bS^T$, we have $\bR^+=\bS$.
Now we show that $\bS$ satisfies {\bf P.\ref{prop1}}-{\bf P.\ref{prop2}}. \\

\paragraph{Our Encoding Satisfies {\bf P.\ref{prop1}}.}
We need to show a map $f:[p]\to\mathcal{P}([d])$ that satisfies {\bf P.\ref{prop1}}.
Let us define the function $f$ as follows, where $(q=m-2t)$ and $p=\lceil\frac{d}{q}\rceil$:
\begin{align}\label{eq:function_CD}
f(i) := 
\begin{cases}
[(i-1)*q+1 : i*q] & \text{ if } 1 \leq i < p, \\
[(p-1)*q+1 : d] & \text{ if } i=p,
\end{cases}
\end{align}
and for any $\U\subseteq[p]$, we define $f(\U):=\cup_{i\in\U}f(i)$.
It is clear from the definition of $f$ that {\sf (i)} $\{f(i):i\in[p]\}$ partitions $[d]$; {\sf (ii)} for every $i\in[p-1]$ we have $|f(i)|=q$, and that $|f(p)| \leq q$.
Recall that $q=m-2t$.
For the 3rd property, note that, for any $\U\subseteq[p]$, all the columns of $\bS_{i\U}$ whose indices belong to $[d]\setminus f(\U)$ are identically zero, which implies that we have 
\begin{align}\label{eq:reduced-Sw}
\bS_{i\U}\bw = \bS_{i\U,f(\U)}\bw_{f(\U)},\quad \text{ for every }\bw\in\R^d,
\end{align} 
which in turn implies
that 
\begin{align}\label{eq:reduced-SX}
\bS_{i\U}\bX^T = \bS_{i\U,f(\U)}\bX_{f(\U)}^T.
\end{align}
Since $\bS^+=\bS^T$, we have $\bS_{i\U}^+=\bS_{i\U}^T$ for every $i\in[m]$ and every $\U\subseteq[p]$.
With these, our update rule $\bv_{i\U}^{t+1} = \bS_{i\U}\bw^t - \alpha_t \bS_{i\U}\bX^T\phi'(\bX\bw^t)$\footnote{We emphasize that we used $\bS^+=\bS^T$ crucially to equivalently write our update rule $\bv_{i\U}^{t+1} = \bR_{i\U}^+\bw^t - \alpha\bR_{i\U}^T\bX^T\phi'(\bX\bw^t)$ from \eqref{eq:partial-update} as $\bv_{i\U}^{t+1} = \bS_{i\U}\bw^t - \alpha_t \bS_{i\U}\bX^T\phi'(\bX\bw^t)$. This follows because $\bS^+=\bS^T$ and we take $\bR^+=\bS$, which together imply that $\bR_{i\U}^+=\bR_{i\U}^T=\bS_{i\U}$.}
can equivalently be written as 
\begin{align}\label{eq:reduced-update}
\bv_{i\U}^{t+1} = \bS_{i\U,f(\U)}\bw_{f(\U)}^{t+1}, 
\end{align}
where 
\begin{align}\label{eq:update-on-w}
\bw_{f(\U)}^{t+1}=\bw_{f(\U)}^t - \alpha_t \bX_{f(\U)}^T\phi'(\bX\bw^t).
\end{align} 
Observe that \eqref{eq:update-on-w} is the same update rule as \eqref{eq:update-w}, which implies that if each worker $i$ updates $\bv_{i\U}$ according to the CD update rule, then the collective update at all the worker nodes automatically updates $\bw_{f(\U)}$ according the CD update rule. Now we show that our invariant $\bv^{t+1}=\bS\bw^{t+1}$ is maintained. We show this by induction. Base case $\bv^{0}=\bS\bw^0$ holds by construction. For the inductive case, assume that $\bv^t=\bS\bw^t$ holds at time $t$ and we show $\bv^{t+1}=\bS\bw^{t+1}$ holds at time $t+1$.

Define $\overline{\U}:=[p]\setminus\U$ and $\overline{f(\U)}:=[d]\setminus f(\U)$.
Since we did not update $\bv_{i\barU}^t$'s, we have $\bv_{i\barU}^{t+1}=\bv_{i\barU}^t$ for every $i\in[m]$. This, together with the inductive hypothesis (i.e., $\bv^t=\bS\bw^t$), implies that 
\begin{align}\label{eq:update-barU-tplus1-to-t}
\bv_{i\barU}^{t+1}=\bS_{i\barU}\bw^t.
\end{align} 
Since $f(\barU)=\barfU$, we have from \eqref{eq:reduced-Sw} that 
\begin{align}\label{eq:update-barU}
\bS_{i\barU}\bw^t=\bS_{i\barU,\barfU}\bw_{\barfU}^t.
\end{align} 
It is clear from \eqref{eq:update-on-w} that $\bw_{\barfU}^t$ did not get an update when we updated $\bv_{i\U}^t$'s, which implies that $\bw_{\barfU}^{t+1}=\bw_{\barfU}^t$. Substituting this in \eqref{eq:update-barU} gives $\bS_{i\barU}\bw^t=\bS_{i\barU,\barfU}\bw_{\barfU}^{t+1}$,
which, by \eqref{eq:reduced-Sw}, yields $\bS_{i\barU}\bw^{t}=\bS_{i\barU}\bw^{t+1}$. This, together with \eqref{eq:update-barU-tplus1-to-t}, implies
\begin{align}\label{eq:update-on-barU-t+1}
\bv_{i\barU}^{t+1}=\bS_{i\barU}\bw^{t+1}.
\end{align}
We already have from \eqref{eq:reduced-Sw} and \eqref{eq:reduced-update} that
\begin{align}\label{eq:update-on-U-t+1}
\bv_{i\U}^{t+1}=\bS_{i\U}\bw^{t+1}.
\end{align}
Since \eqref{eq:update-on-barU-t+1} and \eqref{eq:update-on-U-t+1} hold for every $i\in[m]$, we have $\bv^{t+1}=\bS\bw^{t+1}$. Hence, the invariant is maintained.

\paragraph{Our Encoding Satisfies {\bf P.\ref{prop2}}.}
If we let 
\begin{align*}
\bv_{[m]\U} &:= [\bv_{1\U}^T\ \bv_{2\U}^T\hdots\bv_{m\U}^T]^T, \\
\bS_{[m]\U,f(\U)} &:= [\bS_{1\U,f(\U)}^T\ \bS_{2\U,f(\U)}^T\hdots\bS_{m\U,f(\U)}^T]^T,
\end{align*}
then the collective update \eqref{eq:reduced-update} from all the workers can be written as 
\begin{align}\label{eq:reduced-update-on-v}
\bv_{[m]\U}^{t+1} = \bS_{[m]\U,f(\U)}\bw_{f(\U)}^{t+1}.
\end{align}

It is easy to verify that for every choice of $\U\subseteq[p]$, $\bS_{[m]\U,f(\U)}$ is a full column-rank matrix, which implies that 
we can in principle recover the updated $\bw_{f(\U)}^{t+1}$ from $\bv_{[m]\U}^{t+1} = \bS_{[m]\U,f(\U)}\bw_{f(\U)}^{t+1}$.
Now we show that not only can we recover $\bw_{f(\U)}^{t+1}$ from $\{\bS_{i\U,f(\U)}\bw_{f(\U)}^{t+1}\}_{i=1}^m$, but also efficiently recover $\bw_{f(\U)}^{t+1}$ from the {\em erroneous} vectors $\{\bS_{i\U,f(\U)}\bw_{f(\U)}^{t+1}+\be_{i\U}\}_{i=1}^m$, where at most $t$ out of $m$ error vectors $\{\be_{i\U}\}_{i=1}^m$ are non-zero and may have arbitrary entries. 
Let $\U=\{j_1,j_2,\hdots,j_{|\U|}\}$, and for every $i\in[m]$, let $\be_{i\U}=[e_{ij_1} e_{ij_2}\hdots e_{ij_{|\U|}}]^T$.
Master equivalently writes $\{\bS_{i\U,f(\U)}\bw_{f(\U)}^{t+1}+\be_{i\U}\}_{i=1}^m$ as $|\U|$ systems of linear equations.
\begin{equation}\label{eq:subsystems-CD}
\tilde{h}_i(\bw_{f(\U)}^{t+1}) = \btS_{i,f(\U)}\bw_{f(\U)}^{t+1} + \bte_i,\quad i\in\U,
\end{equation}
where, for every $i\in\U$, $\bte_i=[e_{1i},e_{2i},\hdots,e_{mi}]^T$ and 
$\btS_{i,f(\U)}$ is an $m\times|f(\U)|$ matrix whose $j$'th row is equal to the $i$'th row of $\bS_{j\U}$, for every $j\in[m]$. 
Note that at most $t$ entries in each $\bte_i$ are non-zero. 
Observe that $\{\bS_{i\U,f(\U)}\bw_{f(\U)}^{t+1}+\be_{i\U}\}_{i=1}^m$ and $\{\btS_{i,f(\U)}\bw_{f(\U)}^{t+1} + \bte_i\}_{i\in\U}$ 
are equivalent systems of linear equations, and we can get one from the other.
Observe that \eqref{eq:subsystems-CD} is similar to \eqref{eq:subsystems}: $\btS_{i,f(\U)}$ is equal to $\btS_i$ (for the same $i$) with some of its zero columns removed; and adding zero columns to $\btS_{i,f(\U)}$ will not change the value of $\tilde{h}_i(\bw_{f(\U)}^{t+1})$. 
Now, using the machinery developed in \Sectionref{matrix-vector-mult} we can recover $\bw_{f(\U)}^{t+1}$ from \eqref{eq:subsystems-CD} in $O(|\U|m^2)$ time. 

\subsection{Resource Requirement Analysis}\label{subsec:CD-algorithm-analysis}
In this section, first we give our algorithm developed for distributed coordinate descent in the presence of $t$ (out of $m$) adversarial worker nodes, whose pictorial description is given in \Figureref{problem-setup_CD}. 

We use two encoding matrices $\bL\in\R^{(p_1m)\times n}$ and $\bR\in\R^{d\times(p_2m)}$.
Let $\bL=[\bL_1^T\ \bL_2^T\ \hdots\ \bL_m^T]^T$ and $\bR=[\bR_1\ \bR_2\ \hdots\ \bR_m]$, where each $\bL_i$ is a $p_1\times n$ matrix with $p_1=\lceil\frac{n}{m-2t}\rceil$ and each $\bR_i$ is a $d\times p_2$ matrix with $p_2=\lceil\frac{d}{m-2t}\rceil$. 
Worker $i$ stores both $\btX_i^L=\bL_i\bX$ and $\btX_i^R=\bX\bR_i$. 
Roughly, $\bL$ is used to recover $\bX\bw$ from the erroneous $\{\bL_i\bX\bw+\be_i\}_{i=1}^m$, and $\bR$ is used to update the parameter vector reliably despite errors.
Here $\bL$ is a full column-rank matrix and $\bR$ is a full row-rank matrix. 
Initialize with an arbitrary $\bw^0$ and let $\bv^0=\bR^+\bw^0$. Repeat the following until convergence:

\begin{enumerate}[leftmargin = *]
\item At iteration $t$, master sends $(\bw_{f(\U)}^{t-1}-\bw_{f(\U)}^{t})$\footnote{Observe that master need not send the locations $f(\U)$, because workers can compute those by themselves, as they know both $\U$ and the function $f$.} to all the workers (at $t=0$, master sends $\bw_0$), where $\U\subseteq[p_2]$ is the set of indices used for updating $\bv_{i\U}^{t-1}$'s in the previous iteration, which in turn updated $\bw_{f(\U)}^{t-1}$; see \eqref{eq:reduced-update} and \eqref{eq:update-on-w} in \Subsectionref{encoding-decoding}. 

\item Worker $i$ computes $\btX_i^L(\bw_{f(\U)}^{t-1}-\bw_{f(\U)}^{t})=\bL_i\bX(\bw_{f(\U)}^{t-1}-\bw_{f(\U)}^{t})$ and sends it to the master.\footnote{With some abuse of notation, when we write $\bX\bw_{f(\U)}$, we implicitly assume that $\bw_{f(\U)}$ is a length $d$ vector, which has $0$'s in the indices that lie in $\barfU$.}
Upon receiving $\{\btX_i^L(\bw_{f(\U)}^{t-1}-\bw_{f(\U)}^{t})+\be_i\}_{i=1}^m$, where at most $t$ of the $\be_i$'s are non-zero and may have arbitrary entries, the master applies the decoding procedure of \Sectionref{matrix-vector-mult} and recovers $\bX(\bw_{f(\U)}^{t-1}-\bw_{f(\U)}^{t})$. 
We assume that master keeps $\bX\bw^{t-1}$ from the previous iteration (which is equal to $\bzero$ if $t=0$), it can compute $\bX\bw^t = \bX\bw^{t-1} - \bX(\bw_{f(\U)}^{t-1}-\bw_{f(\U)}^{t})$. Note that if $t=0$, this is equal to $\bX\bw^0$.

\item After obtaining $\bX\bw^t$, master computes $\phi'(\bX\bw^t)$, picks a subset $\U\subseteq[p_2]$ (randomly or in a round robin fashion to cover $[p_2]$ in a few iterations), and sends $(\phi'(\bX\bw^t),\U)$ to all the workers. 

\item Each worker node $i\in[m]$ updates $\bv_{i\U}^{t+1}\leftarrow \bv_{i\U}^t - \alpha_t\nabla_{i\U}\phi(\btX\bv^t)=\bv_{i\U}^t - \alpha_t(\btX_{i\U}^R)^T\phi'(\bX\bw^t)$, while keeping the other coordinates of $\bv_i^t$ unchanged. Worker $i$ sends $\bv_{i\U}^{t+1}$ to the master. Note that $\bv_{i\U}^{t+1}=\bR_{i\U,f(\U)}^+\bw_{f(\U)}^{t+1}$, where $\bw_{f(\U)}^{t+1}=[\bw_{f(\U)}^t - \alpha \bX_{f(\U)}^T\phi'(\bX\bw^t)]$; see \eqref{eq:reduced-update} and \eqref{eq:update-on-w} in \Subsectionref{encoding-decoding}.

\item Upon receiving $\{\bv_{i\U}^{t+1}+\be_{i\U}\}_{i=1}^m$, where at most $t$ of the $\{\be_{i\U}\}_{i=1}^m$'s are non-zero and may have arbitrary entries, master applies the decoding procedure (since our encoding satisfies {\bf P.\ref{prop2}}) and recovers $\bw_{f(\U)}^{t+1}$.
\end{enumerate}

Now we analyze the total amount of resources (storage, computation, and communication) required by the above algorithm
and prove \Theoremref{main-result_CD}.
Fix an $\epsilon>0$. Let the corruption threshold $t$ satisfy $t\leq \lfloor(\epsilon/(1+\epsilon))\cdot(m/2)\rfloor$.

\subsubsection{Storage Requirement:}
By a similar analysis done in \Subsectionref{resource_GD}, we can show that the total storage at all worker nodes is roughly equal to $2(1+\epsilon)|\bX|$. 

\subsubsection{Computational Complexity:}
We can divide the computational complexity of our scheme as follows:
\begin{itemize}[leftmargin = *]
\item{\it Encoding the data matrix.} 
By a similar analysis done in \Subsectionref{resource_GD}, we can show that the total encoding time is $O\left((\frac{\epsilon}{1+\epsilon}m+1)nd\right)$.
Note that this encoding is to be done only once.

\item{\it Computation at each worker node.} 
Suppose that in each iteration of our algorithm, all the workers update $\tau$ coordinates of $\bv_i$'s.
Fix an iteration $t$ and assume that at iteration $(t-1)$, workers updated the coordinates in the set $\U\subseteq[p_2]$, where $|\U|=\tau$.
Recall from {\bf P.\ref{prop1}} that updating $\tau=|\U|$ coordinates of each $\bv_i^{t-1}$ automatically updates $\bw_{f(\U)}^{t-1}$.
Upon receiving $(\bw_{f(\U)}^{t-1}-\bw_{f(\U)}^{t})$ from the master node, each worker $i$ computes $\btX_i^\bL(\bw_{f(\U)}^{t-1}-\bw_{f(\U)}^{t})$, 
and reports back the resulting vector. Note that $(\bw_{f(\U)}^{t-1}-\bw_{f(\U)}^{t})$ has at most $|f(\U)|=\frac{\tau m}{1+\eps}$ non-zero elements,
which together with that $\btX_i^\bL$ is a $p_1\times d$ matrix, implies that computing $\btX_i^\bL(\bw_{f(\U)}^{t-1}-\bw_{f(\U)}^{t})$ takes $O(p_1\cdot|f(\U)|)=O(n\tau)$ time.\footnote{Note that in the very first iteration, master sends $\bw^0$, which may be a dense length $d$ vector, and computing $\btX_i\bL\bw^0$ at the $i$'th worker can take $O(p_1d)=O((1+\eps)\frac{nd}{m})$ time. This is only for the first iteration.}
In the second round, given $\phi'(\bX\bw^t)$, since $(\btX_{i\U}^R)^T$ is of dimension $n\times \tau$, updating $\bv_{i\U}^{t}$ requires $O(n\tau)$ time, where $\tau=|\U|$. So, the total time taken by each worker is $O(n\tau)$.

\item{\it Computation at the master node.}
Once master receives $\{\bL_i\bX(\bw_{f(\U)}^{t-1}-\bw_{f(\U)}^{t})+\be_i\}_{i=1}^m$, applying the decoding procedure of \Sectionref{matrix-vector-mult} to obtain $\bX(\bw_{f(\U)}^{t-1}-\bw_{f(\U)}^{t})$ from these erroneous vectors requires $O(p_1m^2)=O((1+\epsilon)nm)$ time. After that obtaining $\bX\bw^t$ takes another $O(n)$ time.
Given $\bX\bw^t$, computing $\phi'(\bX\bw^t)$ takes $O(n)$ time, assuming that computing $\ell'(\langle \bx_i,\bw^t\rangle;y_i)$ requires unit time, where $\langle \bx_i,\bw^t\rangle$ is equal to the $i$'th entry of $\bX\bw^t$.
Upon receiving $\{\bv_{i\U}^{t+1}+\be_{i\U}\}_{i=1}^m$, where $\bv_{i\U}^{t+1}=\bR_{i\U,f(\U)}^+\bw_{f(\U)}^{t+1}$, for all $i\in[m]$, recovering $\bw_{f(\U)}^{t+1}$ requires $O(\tau m^2)$ time. 
So, the total time taken by the master node is $O((1+\epsilon)nm+\tau m^2)$. 
\end{itemize}

\subsubsection{Communication Complexity:}
Suppose workers update $\tau$ coordinates of $\bv_i$'s in each iteration. Then
{\sf (i)} master broadcasts $\left(\frac{\tau m}{1+\eps}+n\right)$ real numbers, $\frac{\tau m}{1+\eps}$ in the first round to represent $\bw_{f(\U)}^t$ and $n$ in the second round to represent $\phi'(\bX\bw^t)$; and
{\sf (ii)} each worker sends $\left(\tau+(1+\eps)\frac{n}{m}\right)$ real numbers, $(1+\eps)\frac{n}{m}$ in the first round for computing $\bX\bw^t$ at the master node and $\tau$ in the second iteration to represent $\bv_{i\U}^t$.

%% file: extensions.tex
\section{Extensions}\label{sec:extensions}
In this section, we give a few important extensions of our coding scheme developed earlier in \Sectionref{matrix-vector-mult}.
First we give a Byzantine-resilient and communication-efficient method for stochastic gradient descent (SGD).
Second we show how to exploit the specific structure of our encoding matrix to efficiently extend our coding technique to the streaming data model.
In the end, we give a few more important applications, where our method can be applied constructively.

\subsection{Stochastic Gradient Descent}\label{subsec:solution_SGD}
Stochastic gradient descent (SGD) \cite{RobbinsMonro51} is another alternative if full gradients are too costly to compute.
In each iteration of SGD, we sample a data point uniformly at random, compute a gradient on that sample, and update the parameter vector based on that.
We start with an arbitrary/random parameter vector $\bw_0\in\R^d$ and update it according the following update rule:
\begin{align}\label{eq:sgd-update-rule}
\bw_{t+1} = \bw_t - \alpha_t\nabla f_{r_t}(\bw_t), \quad t=1,2,3,\hdots
\end{align}
where $r_t$ is sampled uniformly at random from $\{1,2,\hdots,n\}$.
This ensures that the expected value of the gradient is equal to the true gradient.
Due to its simplicity and remarkable empirical performance, SGD has become arguably the most widely-used optimization algorithm in many large-scale applications, especially in deep learning \cite{Bottou10,RakhlinShSr12,Dean_deeplearning12}.
We want to run SGD in a distributed setup, where data is distributed among $m$ worker nodes and at most $t$ of them can be corrupt;
see \Subsectionref{adversary-model} for details on our adversary model.

\paragraph{Our solution.}
In the plain SGD, we sample a data point randomly and compute its gradient. So, we give a method in which, at any iteration $t$, master picks a random number $r_t$ in $\{1,2,\hdots,n\}$, broadcasts it, and recovers the $r_t$'th data point $\bx_{r_t}$. 
Once the master has obtained $\bx_{r_t}$, it can compute a gradient on it and updates the parameter vector. 
Since master recovers the data points, we can optimize for non-convex problems also; essentially, we could optimize anything that the plain SGD can.
Our method is described below.

We encode $\bX^T$ using the $\left\lceil d/(m-2t) \right\rceil\times d$ encoding matrix $\bS^{(2)}$, which has been defined in \Subsectionref{resource_GD}. For simplicity, we denote $\bS^{(2)}$ by $\bS$. Let $\bS=[\bS_1^T\ \bS_2^T\ \hdots\ \bS_m^T]^T$. 
Note that the $j$'th worker stores $\bS_j\bX^T$. 
Let $\btX:=\bS\bX^T$, which is a $\left\lceil d/(m-2t) \right\rceil\times n$ matrix, whose $i$'th column is the encoding $\btx_i := \bS\bx_i$ of the $i$'th data point $\bx_i$. 
Using the method developed in \Sectionref{matrix-vector-mult}, given $\{\bS_j\bx_i+\be_j\}_{j=1}^m$, where $\be_j=\bzero$ if the $j$'th worker is honest, otherwise can be arbitrary, master can recover $\bx_i$ exactly in $O((1+\epsilon)md)$ time. Our main theorem is stated below, a proof of which trivially follows from \Sectionref{matrix-vector-mult}.
\begin{thm}[Stochastic Gradient Descent]\label{thm:main-result_SGD}
Let $\bX\in\R^{n\times d}$ denote the data matrix. Let $m$ denote the total number of worker nodes.
We can compute a stochastic gradient in a distributed manner in the presence of $t$ corrupt worker nodes and $s$ stragglers, with the following guarantees, where $\epsilon>0$ is a free parameter.
\begin{itemize}
\item $(s+t)\leq\left\lfloor\frac{\epsilon}{1+\epsilon}\cdot \frac{m}{2}\right\rfloor$. 
\item Total storage requirement is roughly $(1+\epsilon)|\bX|$. 
\item Computational complexity for each stochastic gradient computation: 
\begin{itemize}
\item at each worker node is $O((1+\epsilon)\frac{d}{m})$.
\item at the master node is $O((1+\epsilon)dm)$.
\end{itemize}
\item Communication complexity for each stochastic gradient computation:
\begin{itemize}
\item each worker sends $\left((1+\epsilon)\frac{d}{m}\right)$ real numbers.
\item master broadcasts $\lceil\log n\rceil$ {\em bits}.
\end{itemize}
\item Total encoding time is $O\left(nd\left(\frac{\epsilon}{1+\epsilon}m+1\right)\right)$.
\end{itemize}
\end{thm}
Observe the distributed gain of our method in the communication exchanged between the workers and the master:
{\sf (i)} master only broadcasts an index in $\{1,2,\hdots,n\}$, which only takes $\lceil\log n\rceil$ {\em bits}; and 
{\sf (ii)} each worker sends roughly $\frac{1+\eps}{m}$ fraction of the total dimension $d$. 
Hence, this method is particularly useful in distributed settings with communication-constrained and band-limited links.
The Remarks \ref{remark:errors-erasures}, \ref{remark:parameters}, \ref{remark:optimal-threshold} are also applicable for \Theoremref{main-result_SGD}.

\begin{remark}[One-round vs.\ two-round approach]
Unlike the two-round approach taken for gradient computation in PGD and also for CD, we give a one-round approach for each iteration of SGD. This is because in each SGD iteration we need to compute the gradient on {\em only one} data point (not the entire dataset, as in the case for each PGD iteration). Because of this, recovering a (random) data-point itself at the master and then computing a gradient on it locally (which is what we do) would be far more efficient than computing gradient on a single data point in a distributed manner. This is in contrast to each gradient computation for PGD, which requires computation of the full gradient (which is the summation of gradients on all $n$ data points). In principle, we can use the one-round for each PGD iteration also in which first we recover all the $n$ data points at master and then compute the full gradient locally, but this approach would defeat the purpose of distributed computation both in terms of storage and computational complexity. Note that our two-round approach for PGD is significantly more efficient than this.

The reason behind taking the two-round approach for CD is because in order to update the local parameter vectors in the $t$'th iteration, workers need access to the MV multiplication $\btX^R\bv^t=\bX\bw^t$ (see the paragraph after \eqref{eq:partial-update} for more details), and in order to provide that we use an extra round -- the first round is used for computing $\bX\bw$ and the second round is used for updating the local parameter vectors.
Again, for CD also, we could adopt a one-round approach where master recovers all the $n$ data points and then do the parameter update, but that would be highly inefficient and defeat the purpose of distributed computation.

One of the main advantages of the one-round approach for SGD is that since we are recovering the data point itself at the master, we can use it to optimize any function, both convex and non-convex. This is in contrast to the two-round approach, which can only be used for generalized linear models only.
\end{remark}

\subsection{Encoding in The Streaming Data Model}\label{subsec:streaming-result}
An attractive property of our encoding scheme is that it is very easy to update with new data points. 
More specifically, our encoding requires the same amount of time, irrespective of whether we get all the data at once, 
or we get each sample point one by one, as in the online/streaming model.
This setting encompasses a more realistic scenario, in which we design our coding scheme with the initial set of data points and distribute the encoded data among the workers. Later on, when we get some more samples, we can easily incorporate them into our existing encoded data. 
We show that updating $(m-2t)$ new data points in $\R^d$ requires $O\left((m-2t)\left((2t+1)d\right)\right)$ time in total, i.e., $O\left((2t+1)d\right)$ amortized-time per data point. 
This is the best one can hope for, since the offline encoding of $n$ data points requires $O\left((2t+1)nd\right)$ total time. 
At the end of the update, the final encoded matrix that we get is the same as the one we would have got had we had all the $n+1$ data points in the beginning. Therefore, the decoding is not affected by this method at all. 
Note that we use the same encoding matrices both for gradient computation as well as for coordinate descent.
So, it suffices to prove our result in the streaming model for any one of them, and we show it for gradient computation below.
\begin{thm}\label{thm:streaming-result}
The total time complexity in encoding all the data points at once, i.e., when encoding is done offline, is the same as the total time complexity in encoding the data points one by one as they come in the streaming model, i.e., when encoding is done online. 
\end{thm}
\begin{proof}
Let $\bS^{(1)}$ and $\bS^{(2)}$ denote the encoding matrices for encoding $\bX$ and $\bX^T$, respectively; see \Subsectionref{encoding-matrix}.
For convenience, we copy over the corresponding encoding matrices $\bS_i^{(1)}$ and $\bS_i^{(2)}$ from \eqref{eq:encoding-matrix-S_i} for the $i$'th worker node in \Figureref{encoding-matrix-S_i}.
\begin{figure}[h]
\begin{subfigure}{0.5\textwidth}\centering
\[
\bS_i^{(1)}=
\begin{bmatrix}
b_{1i} \hdots b_{qi} & & & \\
& \ddots & & \\
& & b_{1i} \hdots b_{qi} & \\
& & & b_{1i} \hdots b_{l_1i}
\end{bmatrix}_{p_1\times n}
\]
\caption{}
\label{fig:encoding-matrix-S1_i}
\end{subfigure}
\begin{subfigure}{0.5\textwidth}\centering
\[
\bS_i^{(2)}=
\begin{bmatrix}
b_{1i} \hdots b_{qi} & & & \\
& \ddots & & \\
& & b_{1i} \hdots b_{qi} & \\
& & & b_{1i} \hdots b_{l_2i}
\end{bmatrix}_{p_2\times d}
\]
\caption{}
\label{fig:encoding-matrix-S2_i}
\end{subfigure}
\caption{\Figureref{encoding-matrix-S1_i} depicts the encoding matrix for the $i$'th worker node for encoding $\bX$, which is used in the first round of the gradient computation. Here $p_1 = \lceil n/q\rceil$, where $q=(m-k)$ and $k$ is equal to the number of rows in the error recovery matrix $\bF$ in \eqref{eq:vandermonde}, and $l_1 = n-(p_1-1)q$. \Figureref{encoding-matrix-S2_i} depicts the encoding matrix for the $i$'th worker node for encoding $\bX^T$, which is used in the second round of the gradient computation. Here $p_2 = \lceil d/q\rceil$ and $l_2 = d-(p_2-1)q$. All the unspecified entries in both the matrices are zero.}
\label{fig:encoding-matrix-S_i}
\end{figure}

Suppose at some point of time we have encoded $n$ data points each lying in $\R^d$ and distributed the encoded data among the $m$ worker nodes. Now a new data sample $\bx\in\R^d$ comes in. We will show how to incorporate it in the existing scheme in $O\left((2t+1)d\right)$ time on average. 

\paragraph{Updating the encoding matrices.}
Fix an arbitrary worker $i\in[m]$.
Note that the new data matrix $\bX$ has dimension $(n+1)\times d$. 
So, the new encoding matrix $\bS_i^{(1)}$ should have $(n+1)$ columns, and we have to add one more column to $\bS_i^{(1)}$. By examining the repetitive structure of $\bS_i^{(1)}$, it is obvious which column to add: 
if $l_1 < q$, then we add the $p_1$-dimensional vector $[0,0,\hdots,0,b_{(l_1+1)i}]^T$ as the last column; otherwise,
if $l_1 = q$, then we add the $(p_1+1)$-dimensional vector $[0,0,\hdots,0,b_{1i}]^T$ as the last column.
In the second case, the number of rows of $\bS_i^{(1)}$ increases by one -- the last row has all zeros, except for the last element, which is equal to $b_{1i}$.
Note that $\bS_i^{(2)}$ does not change at all. 
Observe that if the $i$'th worker performs this update, then it does not have to store its entire encoding matrix $\bS_i^{(1)}$, it only needs to store $n$, $q=(m-k)$, and the $q$ real numbers $b_{1i},b_{2i},\hdots,b_{qi}$, where $q=m-k$, which could be much smaller as compared to $n$ and $d$, and are enough to define $\bS_i^{(1)}$ and $\bS_i^{(2)}$.
\\

\paragraph{Updating the encoded data.}
Now we show how to update the encoded data with the new sample $\bx$. We need to update both $\bS_i^{(1)}\bX$ as well as $\bS_i^{(2)}\bX^T$ for every worker $i\in[m]$.
\begin{itemize}[leftmargin = *]
\item \emph{Updating $\bS_i^{(1)}\bX$.} If $l_1< q $, then we add $b_{(l_1+1)i}\bx^T$
to the last row of $\bS_i^{(1)}\bX$; otherwise, if $l_1= q $, then we add $b_{1i}\bx$ as a new row in $\bS_i^{(1)}\bX$.
In the first case, the resulting matrix still has $p_1$ rows, whose first $p_1-1$ rows are same as before, and the last row is the sum of the previous row and $b_{(l_1+1)i}\bx^T$. In the second case, the resulting matrix has $(p_1+1)$ rows, whose first $p_1$ rows are the same as before and the last row is equal to $b_{1i}\bx^T$. Note that each row of $\bS_i^{(1)}$ for $i\leq 2t$, has at most $(m-2t)$ non-zero elements; whereas, for $i>2t$, each row of $\bS_i^{(1)}$ has exactly one non-zero entry. Since there are $p_1=\left\lceil n/(m-2t)\right\rceil$ rows in each $\bS_i^{(1)}$, updating $\bS_i^{(1)}\bX$ for every $i\leq 2t$ takes $O(d)$ time; and for $i>2t$, update in $\bS_i^{(1)}\bX$ happens only once in $(m-2t)$ new data points (whenever the second case occurs and the resulting $\bS_i^{(1)}$ has $(p_1+1)$ rows). 
So, updating $(m-2t)$ data points at all $m$ worker nodes require $O\left(2t*(m-2t)d + (m-2t)*d\right)=O((m-2t)(2t+1)d)$ time, i.e., $O\left((2t+1)d\right)$ time per data point.

\item \emph{Updating $\bS_i^{(2)}\bX^T$.} Note that $\bX^T$ is a $d\times(n+1)$ matrix whose last column is equal to the new data sample $\bx$.
Now, to update $\bS_i^{(2)}\bX^T$, we add $\bS_i^{(2)}\bx$ as an extra column. The resulting matrix is of size $p_2\times(n+1)$, whose first $n$ columns are the same as before and the last column is equal to $\bS_i^{(2)}\bx$. 
Since total number of non-zero entries in $\bS_i^{(2)}$ is equal to $d$ if $i\leq 2t$ and equal to $p_2=\left\lceil d/(m-2t)\right\rceil$ if $i>2t$,
the total time required to update a new data point is $O(2t*d + (m-2t)*p_2)=O\left((2t+1)d\right)$.
\end{itemize}
Observe that at the end of this local update at each worker node, the final encoded matrix that we get is the same as the one we would have got had we had all the $n+1$ data points in the beginning. The decoding is not affected by this method at all. 
This completes the proof of \Theoremref{streaming-result}.
\end{proof}

\begin{remark}[Updating the encoded data efficiently with new features]\label{remark:streaming-features} 
Observe that since we encode both $X$ and $X^T$ in an analogous fashion, it follows by symmetry that we can not only update efficiently upon receiving a new data sample, but can also update efficiently if we decide to enlarge the dimension $d$ of each of the $n$ data samples at some point of time -- maybe we figure out some new features of the data to get a more accurate model to overcome under-fitting. In these situations, we don't need to encode the entire dataset all over again, just a simple update is enough to incorporate the changes.
\end{remark}

\begin{remark}[What allows our encoding to be efficient for streaming data?]\label{remark:structured-matrix-implies-streaming}
The efficient update property of our coding scheme is made possible by the repetitive structure of our encoding matrix (see \Figureref{encoding-matrix-S_i}), together with the fact that this structure is independent of the number of data points $n$ and the dimension $d$ -- it only depends on the number of worker nodes $m$ and the corruption threshold $t$.
We remark that other data encoding methods in literature, even for weaker models, do not support efficient update. 
For example, the encoding of \cite{KarakusSuDiYi17}, which was designed for mitigating stragglers, depends on the dimensions $n$ and $d$ of the data matrix. So, it may not efficiently update if a new data point comes in.
\end{remark}

\subsection{More Applications.}\label{subsec:applications}
There are many iterative algorithms, other than the gradient descent for learning GLMs, which use repeated MV multiplication. 
Some of them include 
{\sf (i)} the power method for computing the largest eigenvalue of a diagonalizable matrix, which is used in Google's PageRank algorithm \cite{IpsenWi06}, Twitter's recommendation system \cite{GuptaGoLiShWaZa13}, etc.; 
{\sf (ii)} iterative methods for solving sparse linear systems \cite{Saad03};
{\sf (iii)} many graph algorithms, where the graph is represented by a fixed adjacency matrix, \cite{KepnerGiGraphLABook11}. 
In large-scale implementation of these systems, where Byzantine faults are inevitable, the method described in this
paper can be of interest.

In most of these applications, the underlying matrix $\bA$ is generally sparse, which is exploited to gain computational
efficiency. So, it is desired not to lose sparsity even if we want resiliency against Byzantine attacks.
Fortunately, our encoding matrix $\bS$ is sparse (see \eqref{eq:encoding-matrix-S_i}), which ensures that
the encoded matrix $\bS\bA$ will not lose the sparsity of $\bA$:
Each of the first $pk$ rows of $\bS$ has at most $(m-k)$ (where $k=2t$) non-zero elements,
and each of the remaining rows has exactly one 1.
Since $m$ is the number of worker nodes, which may be small,
and we can take $t$ to be up to $\lfloor\frac{m-1}{2}\rfloor$, we may have a few non-zero entries in each row of $\bS$ (in the extreme case when $2t=m-1$, each row of $\bS$ has only one non-zero entry).
In a sense, we are getting Byzantine-resiliency almost for free without compromising the computational efficiency that is made possible by the sparsity of the matrix.

%% file: experiments.tex
\section{Numerical Experiments}\label{sec:experiments}
In this section, we validate the efficacy of our proposed methods by numerical experiments.
We run distributed gradient descent (GD) and coordinate descent (CD) for linear regression $\arg\min_{\bw\in\R^d}\|\bX\bw-\by\|_2^2$. 
As mentioned in \Subsectionref{setting-pgd}, for linear regression (which is equal to ridge regression when $h=0$), the projected gradient descent (PGD) reduces to gradient descent (GD).
Since we are doing exact computation (computing the gradients exactly in the case of GD and updating the coordinates exactly in the case of CD),
{\sf (i)} there is no need to check the convergence, and
{\sf (ii)} our algorithm will perform exactly the same whether we are working with synthetic datasets or real datasets, hence, we will work with a synthetic dataset. 
We run our algorithms\footnote{We implement our algorithm in Python, and run it on an iMac machine with 3.8 GHz Quad-Core Intel Core i5 processor and 16 GB 2400 MHz DDR4 memory.}
 with $m=15$ worker nodes on two datasets: $(n=10,000, d=250)$ and $(n=20,000, d=22,000)$.
For both the datasets, we generate $(\bX,\by)$ by sampling $\bX\leftarrow\N(0,I)$ and $\by = \bX\theta + \bz$, 
where $\theta\in\R^d$ has $d/3$ non-zero entries, all of them are i.i.d.~according to 
$\N(0,4)$, 
and each entry of $\bz\in\R^n$ is sampled from $\N(0,1)$ i.i.d. 
In each round of the gradient computation, the adversary picks $t$ worker nodes uniformly at random,
and adds independent random vectors of appropriate length as errors, 
whose entries are sampled from $\N(0,\sigma^2)$ i.i.d.~with $\sigma=100$, 
to the true vectors. 
\begin{figure}[htb]
\centering
\includegraphics[scale=0.5]{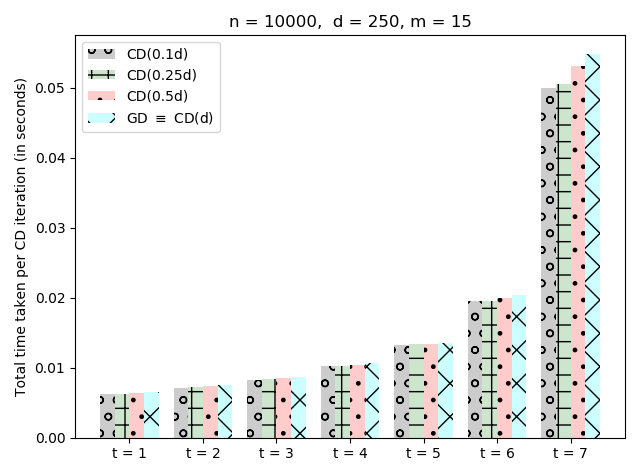}
\caption{We run our algorithms (CD and GD) with $15$ worker nodes on a dataset with $n=10,000,d=250$.
This plot reports how the total time taken (in seconds) for updating different number of coordinates in each CD iteration changes with varying number of corrupt worker nodes from $t=1$ to $t=7$. In the figure, we plot the total time taken (per iteration) for updating the $\gamma$-fraction of $d$ coordinates for $\gamma=0.1,0.25,0.5,1$.
Note that CD with $\gamma=1$ is equivalent to full gradient computation as in GD.}
\label{fig:total-time}
\end{figure}

\begin{figure*}[htb]
\begin{center}
\begin{tabular}{| p{0.06\textwidth} | p{0.08\textwidth} | p{0.08\textwidth} | p{0.08\textwidth} | p{0.08\textwidth} | p{0.08\textwidth} | p{0.08\textwidth} | p{0.08\textwidth} | p{0.08\textwidth} |}
\hline
& \multicolumn{2}{c|}{\ CD($0.1d$)} & \multicolumn{2}{c|}{\ CD($0.25d$)} &\multicolumn{2}{c|}{\ CD($0.5d$)} &\multicolumn{2}{c|}{\ GD $\equiv$ CD($d$)}\\
\hline
&\ Worker &\ Master &\ Worker &\ Master &\ Worker &\ Master &\ Worker &\ Master\\ 
\hline
\ $t=1$ &\ 0.0020 &\	0.0120	&\ 0.0073 &\	0.0182	&\ 0.0122 &\ 0.0199 &\ 0.0493 &\ 0.0214\\
\ $t=2$ &\ 0.0044 &\	0.0187	&\ 0.0092	&\ 0.0212 &\ 0.0188	&\ 0.0277 &\ 0.0953	&\ 0.0393\\
\ $t=3$ &\ 0.0054 &\	0.0201	&\ 0.0118 &\ 0.0242 &\ 0.0269 &\ 0.0324 &\ 0.1213&\ 0.0561\\
\ $t=4$ &\ 0.0063 &\	0.0253	&\ 0.0159 &\ 0.0327 &\ 0.0488 &\ 0.0468 &\ 0.1602 &\ 0.0610\\
\ $t=5$ &\ 0.0107 &\ 0.0342 &\ 0.0328 &\ 0.0460 &\ 0.0776 &\ 0.0738 &\ 0.2943 &\ 0.0826\\
\ $t=6$ &\ 0.0205 &\ 0.0717 &\ 0.0764 &\ 0.0833 &\ 0.1330 &\ 0.1088 &\ 0.8929 &\ 0.1227 \\ \hline
\end{tabular}
\end{center}
\caption{We run our algorithms (CD and GD) with $15$ worker nodes on a dataset with $n=20,000,d=22,000$,
and separately report the maximum time taken by any single worker and the master per iteration 
against varying number of corrupt worker nodes from $t=1$ to $6$. For CD, we run our algorithm for updating different number of coordinates. The first two columns correspond to the case when updating $0.1$-fraction of $d$ coordinates, the next two columns for $0.25$-fraction, and so on. The last two columns correspond to updating all the coordinates, which is equivalent to full gradient computation as in GD.}
\label{fig:separate-time}
\end{figure*}

\subsection{$n=10,000, d=250, m=15$}
In \Figureref{total-time}, we plot the total time taken (which is the sum of the maximum time taken by any single worker node and the time taken by the master node in both rounds) 
for updating different number of coordinates in one CD iteration, with varying number of corrupt worker nodes from $t=1$ to $t=7$. We plot the time needed for updating $\gamma$-fraction of $d$ coordinates for four different values of $\gamma$ (i.e., $\gamma=0.1,0.25,0.5,1$) and we denote it by CD$(\gamma d)$ for $\gamma=0.1,0.25,0.5,1$. Recall that CD$(d)$ is equivalent to full gradient computation as in the case of GD. Note that, 
when $t=7$, we have $\epsilon=m-1$, which is the main cause behind the significant increment in time for $t=7$.

\subsection{$n=20,000, d=22,000, m=15$}
In \Figureref{separate-time}, we report separately, the maximum time taken by any single worker node and the time taken by the master node (together in both the rounds) in one CD iteration for updating different number of coordinates and also for GD, with varying number of corrupt worker nodes from $t=1$ to $t=6$. As in the above case, we report the time needed for updating $\gamma$-fraction of $d$ coordinates for four different values of $\gamma$. Observe that the time taken by the master node is orders of magnitude less than the time taken by the worker nodes. We can also observe that with the running time in a worker node per iteration for CD($0.1d$) is 95\% less than that for GD, while this time saving in the master node is more than 40\%.